\newcommand{\CLASSINPUTtoptextmargin}{1.75cm}
\newcommand{\CLASSINPUTbottomtextmargin}{1.75cm}
\newcommand{\CLASSINPUToutersidemargin}{1.7cm}

\documentclass[10pt,twocolumn,journal]{IEEEtran}

\usepackage{amsmath, mathrsfs, mathtools}
\usepackage{amsfonts}
\usepackage{amssymb}
\usepackage{graphicx}
\usepackage{color}
\usepackage{multirow}
\usepackage{graphicx}
\usepackage{stmaryrd}

\allowdisplaybreaks

\usepackage{amsfonts}
\usepackage{amssymb}
\usepackage{graphicx}
\usepackage{color}
\usepackage{multirow}
\usepackage{graphicx}

\newcommand{\subparagraph}{}
\usepackage[compact]{titlesec}
\titlespacing*{\section}{15pt}{1.2\baselineskip}{0.9\baselineskip}

\allowdisplaybreaks

\newcommand{\myhash}{%
  {\settoheight{\dimen0}{C}\kern-.05em\, \resizebox{!}{\dimen0}{\raisebox{\depth}{\#}}}}

\usepackage{caption}

\usepackage{subcaption}
\usepackage[numbers,sort&compress]{natbib}

\usepackage{algorithm}
\usepackage{algpseudocode}
\usepackage{pifont}
\usepackage{varwidth}

\def\betam{{\boldsymbol{\beta}}}
\def\gammam{{\boldsymbol{\gamma}}}

\def\probd{\bbp_{\sfD}}
\def\probf{\bbp_{\sfF\sfA}}

\newcommand{\vc}[1]{\mathbf{#1}}
\newcommand{\mt}[1]{\mathbf{#1}}
\newcommand{\fro}{\text{F}}

\newcommand{\SigmaEmp}{\widehat{\Sigmam}_\bfy} 

\newcommand{\sigmaEmp}{\widehat{\sigmam}_\bfy} 

\newcommand{\Amatrix}{\bA}  

\newcommand{\Aop}{\mathcal{A}}  

\newcommand{\Acol}{\mt{A}} 

\newcommand{\avec}{\bfa} 

\newcommand{\bxtrue}{\gammam^\circ} 
\newcommand{\xtrue}{\gamma^\circ} 

\newcommand{\xsome}{\gammam} 

\newcommand{\bxnnls}{\gammam^*} 

\newcommand{\ID}{\bI} 

\newcommand{\dimPilots}{D_c} 

\newcommand{\dimParam}{{K_c}} 

\newcommand{\ellp}[1]{{#1}} 

\usepackage{multicol}


\usepackage{mathbbol}

\def\mindex#1{\index{#1}}

\def\sq{\hbox{\rlap{$\sqcap$}$\sqcup$}}
\def\qed{\ifmmode\sq\else{\unskip\nobreak\hfil
\penalty50\hskip1em\null\nobreak\hfil\sq
\parfillskip=0pt\finalhyphendemerits=0\endgraf}\fi\medskip}

\long\def\defbox#1{\framebox[.9\hsize][c]{\parbox{.85\hsize}{%
\parindent=0pt
\baselineskip=12pt plus .1pt      
\parskip=6pt plus 1.5pt minus 1pt 
 #1}}}

\long\def\beginbox#1\endbox{\subsection*{}%
\hbox{\hspace{.05\hsize}\defbox{\medskip#1\bigskip}}%
\subsection*{}}

\def\endbox{}

\def\diag{{\text{diag}}}

\def\tr{\mathsf{tr}}

\newsavebox{\junk}
\savebox{\junk}[1.6mm]{\hbox{$|\!|\!|$}}

\def\argmin{\mathop{\rm arg\, min}}

\def\Re{\field{R}}

\def\bA{{\mathbb A}}

\def\bC{{\mathbb C}}
\def\bD{{\mathbb D}}
\def\bE{{\mathbb E}}

\def\bI{{\mathbb I}}

\def\bR{{\mathbb R}}

\def\bbp{{\mathbb p}}

\def\bfA{{\bf A}}

\def\bfC{{\bf C}}

\def\bfF{{\bf F}}

\def\bfH{{\bf H}}
\def\bfI{{\bf I}}

\def\bfN{{\bf N}}

\def\bfP{{\bf P}}
\def\bfQ{{\bf Q}}
\def\bfR{{\bf R}}

\def\bfU{{\bf U}}
\def\bfV{{\bf V}}

\def\bfX{{\bf X}}
\def\bfY{{\bf Y}}
\def\bfZ{{\bf Z}}

\def\bfa{{\bf a}}

\def\bfd{{\bf d}}
\def\bfe{{\bf e}}

\def\bfg{{\bf g}}
\def\bfh{{\bf h}}

\def\bfu{{\bf u}}
\def\bfv{{\bf v}}
\def\bfw{{\bf w}}
\def\bfx{{\bf x}}
\def\bfy{{\bf y}}
\def\bfz{{\bf z}}

\def\ttc{{\mathtt c}}

\def\tte{{\mathtt e}}

\def\ttn{{\mathtt n}}
\def\tto{{\mathtt o}}

\def\sfA{{\sf A}}

\def\sfD{{\sf D}}

\def\sfF{{\sf F}}

\def\sfH{{\sf H}}

\def\bfmath#1{{\mathchoice{\mbox{\boldmath$#1$}}%
{\mbox{\boldmath$#1$}}%
{\mbox{\boldmath$\scriptstyle#1$}}%
{\mbox{\boldmath$\scriptscriptstyle#1$}}}}

\def\bfmY{\bfmath{Y}}

\def\bfmhhaY{\bfmath{\hhaY}} 
\def\bfmhhaY{\hbox to 0pt{$\widehat{\bfmY}$\hss}\widehat{\phantom{\raise 1.25pt\hbox{$\bfmY$}}}}

\def\til={{\widetilde =}}

\def\clA{{\cal A}}

\def\clC{{\cal C}}

\def\clK{{\cal K}}

\def\clN{{\cal N}}

\def\clP{{\cal P}}
\def\clQ{{\cal Q}}

\def\clS{{\cal S}}

\def\clV{{\cal V}}

 \def\FRAC#1#2#3{\genfrac{}{}{}{#1}{#2}{#3}}

\def\ddtp{{\mathchoice{\FRAC{1}{d^{\hbox to 2pt{\rm\tiny +\hss}}}{dt}}%
{\FRAC{1}{d^{\hbox to 2pt{\rm\tiny +\hss}}}{dt}}%
{\FRAC{3}{d^{\hbox to 2pt{\rm\tiny +\hss}}}{dt}}%
{\FRAC{3}{d^{\hbox to 2pt{\rm\tiny +\hss}}}{dt}}}}

\def\average#1,#2,{{1\over #2} \sum_{#1}^{#2}}

\def\eye(#1){{\bf(#1)}\quad}

\newtheorem{theorem}{{\bf Theorem}}

\newtheorem{remark}{{\bf Remark}}

\def\eq#1/{(\ref{e:#1})}

\newcommand{\beqn}[1]{\notes{#1}%
\begin{eqnarray} \elabel{#1}}

\newcommand{\eeqn}{\end{eqnarray} }

\newcommand{\beq}[1]{\notes{#1}%
\begin{equation}\elabel{#1}}

\newcommand{\eeq}{\end{equation}}

\def\bdes{\begin{description}}
\def\edes{\end{description}}

\newcounter{rmnum}

\newcounter{anum}

{\end{list}}

\def\ass(#1:#2){(#1\ref{#1:#2})}

\def\ritem#1{
\item[{\sf \ass(\current_model:#1)}]
}

\newenvironment{recall-ass}[1]{%
\begin{description}
\def\current_model{#1}}{
\end{description}
}

\usepackage{tikz}
\usepackage{pgfplots,tikz-3dplot}
\usetikzlibrary{intersections, pgfplots.fillbetween}
\pgfplotsset{compat=newest}
\usetikzlibrary{patterns}
\usetikzlibrary{positioning}
\usetikzlibrary{datavisualization}
\usetikzlibrary{datavisualization.formats.functions}
\usetikzlibrary{backgrounds}
\usetikzlibrary{shapes,snakes}

\usepgfplotslibrary{external} 
\usepackage{float}
\usepackage{afterpage}
\usepackage{balance}

\def\herm{{\sfH}}

\def\snr{{\mathsf{snr}}}

\def\sigmam{{\boldsymbol{\sigma}}}

\def\cg{{\clC\clN}} 
\newcommand{\normd}[1]{{\left\vert\kern-0.25ex\left\vert\kern-0.25ex\left\vert #1 
    \right\vert\kern-0.25ex\right\vert\kern-0.25ex\right\vert}}

\def \cone {{\ttc\tto\ttn\tte}}

\long\def\comment#1{}

\newcommand{\xv}{{\bf x}}

\newcommand{\Xm}{{\bf X}}

\newcommand{\Xc}{{\cal X}}

\newcommand{\Gammam}{\boldsymbol{\Gamma}}
\newcommand{\Lambdam}{\boldsymbol{\Lambda}}
\newcommand{\Deltam}{\boldsymbol{\Delta}}
\newcommand{\Sigmam}{\boldsymbol{\Sigma}}

\newcommand{\trace}{{\hbox{tr}}}

\renewcommand{\Re}{{\rm Re}}
\renewcommand{\Im}{{\rm Im}}

\newcommand{\transp}{{\sf T}}
\renewcommand{\vec}{{\rm vec}}

\title{A New Scaling Law for Activity Detection in Massive MIMO
  Systems }
\author{Saeid Haghighatshoar,  Peter Jung, Giuseppe Caire
\thanks{The authors are with the Communications and Information Theory Group, Technische Universit\"{a}t Berlin (\{saeid.haghighatshoar, peter.jung, caire\}@tu-berlin.de).}
\thanks{A short version of this paper was accepted for presentation in the 2018 IEEE International Symposium on Information Theory (ISIT) in Vail, Colorado, USA \cite{AD:isit2018}.}
}

\begin{document}

\maketitle

\begin{abstract}
	
	  In this paper, we study the problem of \textit{activity
		detection} (AD) in a massive MIMO setup, where the Base Station
	(BS) has $M \gg 1$ antennas. We consider a block fading channel model where
	the $M$-dim channel vector of each user remains almost constant over
	a \textit{coherence block} (CB) containing $D_c$ signal dimensions. We
	study a setting in which the  number of potential users $K_c$ assigned to a specific CB is much larger than the dimension of the CB  $D_c$ ($K_c \gg D_c$)
	but at each time slot only $A_c \ll K_c$ of them are active.  Most
	of the previous results, based on compressed sensing, require that 
	$A_c\le D_c$,  which is a bottleneck in massive deployment
	scenarios such as Internet-of-Things (IoT) and Device-to-Device (D2D)  communication.  
	In this paper, we show that one can overcome  this fundamental
	limitation when the number of BS antennas $M$ is sufficiently large.
	More specifically, we derive a
	\textit{scaling law} on the parameters $(M, D_c, K_c, A_c)$ and also
	\textit{Signal-to-Noise Ratio} (SNR) under which our proposed AD scheme succeeds.
	Our analysis indicates that with a CB of dimension $D_c$, and a sufficient number of BS antennas $M$ with $A_c/M=o(1)$,  one
	can identify the activity of 
	$A_c=O(D_c^2/\log^2(\frac{K_c}{A_c}))$ active users, which is much larger than  the previous bound $A_c=O(D_c)$
	obtained via traditional compressed sensing techniques. In particular, in our proposed scheme
	one needs
	to pay only a poly-logarithmic penalty
	$O(\log^2(\frac{K_c}{A_c}))$ for increasing the number of potential
	users $K_c$, which makes it ideally suited for AD in IoT setups.  
	We propose  low-complexity algorithms for AD and provide numerical simulations to
	illustrate our results. We also compare the performance of our proposed AD algorithms with that of other competitive algorithms in the literature.

\end{abstract}


\section{Introduction}
%
Massive connectivity is predicted to play a crucial role in  future generation of wireless cellular networks that support Internet-of-Things (IoT) and Device-to-Device (D2D) communication. In such scenarios, a Base Station (BS) should be able to connect to a large number of devices and the underlying shared communication resources (time, bandwidth, etc.) are dramatically overloaded. However, a key feature of wireless traffic in those systems (especially,  in IoT) is that the device activity patterns are typically sporadic such that over any  communication resource only a small fraction of all potential devices are active. A feasible communication in those scenarios typically consists of two phases: {\bf i})\,identifying the set of active users by spending a fraction of communication resources, {\bf ii})\,serving the set of active users via a suitable scheduling over the remaining communication resources. 

In this paper, we mainly investigate the first phase, namely, identifying the activity pattern of users, known as \textit{activity detection} (AD).  
We consider a generic block fading wireless communication channel between devices and the  BS \cite{tse2005fundamentals}. We assume that the channel can be decomposed into a set of \textit{coherence blocks} (CBs), where each CB consists of $D_c$ signal dimensions over which the channel fading coefficients of each user  remain almost constant, whereas the channel might vary independently across different CBs \cite{tse2005fundamentals}. 
AD is a fundamental challenge in  massive
deployments and random access scenarios to be expected for IoT and D2D (see, e.g., 
\cite{Bockelmann:ETT2013,boljanovic2017user,kueng2016robust} for some recent works). A fundamental 
 limitation when considering solely a single-antenna setting is that the required signal dimension $D_c$ to identify reliably a subset of $A_c$ active users among a set consisting of $K_c$ \textit{potentially active}  users scales as $D_c=O(A_c \log(\frac{K_c}{A_c}))$, thus, almost linearly with $A_c$. To keep
up with the scaling requirements in IoT and D2D setup where $A_c$ and $K_c$ are dramatically large,  it is crucial to
overcome this limitation in an efficient way that does not require devoting too many CBs to AD. 
One of the
recent works along this line is \cite{liu2017massive, chen2018sparse}, where the
authors proposed using multiple antennas at the BS to overcome this fundamental problem. More specifically, they showed that by assigning random Gaussian pilot sequences to the users, one can identify any subset of active users with a vanishing error probability in a massive MIMO setup (when the number of BS antennas $M \to \infty$) \cite{Marzetta-TWC10}  provided that $D_c$ is large. 
%
In contrast, in \cite{docomo} the authors studied another variant of AD and showed that over a CB of dimension $D_c$, one is able to identify up to
$O(D_c^2)$ active users provided that the \textit{Signal-to-Noise Ratio} (SNR) of  the active users is sufficiently large and the number of antennas $M \to \infty$. However, \cite{liu2017massive, chen2018sparse, docomo} do not specify the
finite-length scaling law on the number of antennas, pilot dimension, the number of users, and the number of active users $(M, D_c, K_c, A_c)$ required for a reliable AD.

In this paper, we bridge the gap by deriving
a finite-length scaling law on $(M, D_c, K_c, A_c)$ and SNR for the scheme proposed in \cite{docomo}. In particular, we show that:
{\bf i})\,with a sufficient number of BS antennas, with $A_c/M=o(1)$, 
one can identify reliably the activity of 
$A_c = O(D_c^2/\log^2 (\frac{K_c}{A_c}))$ active users among a set of $K_c$ users, which is orders of magnitude better the bound $A_c=O(D_c/\log (\frac{K_c}{A_c}))$ obtained  previously 
via traditional compressed sensing techniques (see, e.g., \cite{Bockelmann:ETT2013,boljanovic2017user,kueng2016robust}), {\bf ii})\,one needs to pay only a  logarithmic penalty
$O(\log^2(\frac{K_c}{A_c}))$  for increasing the total number of 
users $K_c$. Both features make the proposed scheme very attractive for IoT setups, in which the  number of active users $A_c$ as well as the total number of users  $K_c$ can be extremely large.

%
%
%
%
We  propose several efficient and low-complexity algorithms for AD and show via theoretical analysis as well as numerical simulation that they  fulfill our proposed scaling law. Our proposed AD algorithms depend only on the sample
covariance of the observations obtained at the BS and are robust to statistical variations of the users channels.

\subsection{Notation}
We represent scalar constants by non-boldface letters {(e.g., $x$ or $X$)}, sets by calligraphic letters (e.g., $\Xc$),  vectors by boldface small letters (e.g., $\xv$), and matrices by boldface capital letters (e.g., $\Xm$). 
We denote the $i$-th row and the $j$-th column of a matrix $\bfX$ with the row-vector $\Xm_{i,:}$ and the column-vector $\Xm_{:,j}$ respectively. We denote a diagonal matrix with elements $(s_1, s_2, \dots, s_k)$ by $\diag(s_1, \dots, s_k)$.
We denote the $\ell_p$-norm a vector $\bfx$ and the Frobenius norm of
a matrix $\bfX$ by $\|\bfx\|_{\ellp{p}}$ ($\|\bfx\|=\|\bfx\|_{\ellp{2}}$)
and $\|\bfX\|_\sfF$ respectively.
The  $k \times k$ identity matrix is represented by $\bfI_k$. For an integer $k>0$, we use the shorthand notation $[k]$ for $\{1,2,\dots, k\}$.

\section{Problem Formulation}
\subsection{Signal Model}
We consider a generic block fading wireless channel  between each user and the BS consisting of several CBs with each CB containing $D_c$ signal dimensions over which the channel is almost flat \cite{tse2005fundamentals}, whereas it may change smoothly or almost independently  across adjacent CBs.  
Without loss of generality, we assume that the BS devotes an individual CB to AD of a specific set  of users $\clK_c$ consisting of $K_c:=|\clK_c|$ users. 
To perform AD, the BS assigns a specific pilot sequence to each user  in $\clK_c$, where a pilot sequence for a generic user $k \in \clK_c$ is simply a
vector $\bfa_k =(a_{k,1}, \dots, a_{k,D_c})^\transp \in \bC^{D_c}$ of length $D_c$, which is
transmitted by the user $k$ over $D_c$ signal dimensions in the CB devoted to the AD if it is active. Denoting by $\bfh_k$ the $M$-dim channel vector of the
user $k \in \clK_c$ to $M$  antennas at the BS, we can write the received signal at the BS over the signal dimension $i\in [D_c]$ inside the CB as
\begin{align}\label{sig_eq1}
\bfy[i]=\sum_{k \in \clK_c} b_k\sqrt{g_k} a_{k,i} \bfh_k+ \bfz[i],
\end{align}
where $a_{k,i}$ denotes the $i$-th element of the pilots sequence $\bfa_k$, where $[D_c]:=\{1,\dots, D_c\}$,  where $g_k \in \bR_+$ denotes the large-scale fading coefficient (channel  strength) of the user $k \in \clK_c$, where $b_k \in \{0,1\}$ is a binary variable with $b_k=1$ for active and $b_k=0$ for inactive users, where $\bfz[i]\sim\cg(0, \sigma^2 \bfI_M)$ denotes the additive white Gaussian noise (AWGN) at the $i$-th signal dimension, and where we used the block fading channel model \cite{tse2005fundamentals} where  the channel vector $\bfh_k$ of each user $k \in \clK_c$ is almost constant over the signal dimensions $i \in [D_c]$ inside the CB.
Denoting by $\bfY=[\bfy[1], \dots, \bfy[D_c]]^\transp$ the $D_c \times M$ received signal over $D_c$ signal dimensions and $M$ BS antennas, we can write \eqref{sig_eq1} more compactly as
 \begin{align}\label{pilot_sig}
\bfY=\bfA \Gammam^{\frac{1}{2}} \bfH + \bfZ,
\end{align}
where $\bfA=[\bfa_1, \dots, \bfa_{K_c}]$ denotes the $D_c\times K_c$ matrix of pilot sequences of the users in $\clK_c$, where $\Gammam=\diag(\gammam)$ with $\gammam:=(\gamma_1, \dots, \gamma_{K_c}) \in \bR_+^{K_c}$ and $\gamma_k=b_k g_k$ denotes the channel strengths of the users ($\gamma_k=0$ for inactive ones), and where $\bfH=[\bfh_1, \dots, \bfh_{K_c}]^\transp$ denotes $K_c \times M$  matrix containing the $M$-dim normalized channel vectors of the users. We assume that the channel vectors $\{\bfh_k: k \in \clK_c\}$ are independent from each other and are  spatially white (i.e.,  uncorrelated along the antennas), that is,  $\bfh_k \sim \cg(0, \bfI_M)$. 
We assume that the user pilots  are also normalized to $\|\bfa_k\|^2=D_c$ and define the average SNR of a generic active user $k \in \clK_c$ over $D_c$ pilot dimensions by 
\begin{align}
\snr_k=\frac{\|\bfa_k\|^2 \gamma_k \bE[\|\bfh_k\|^2]}{\bE[\|\bfZ\|_\sfF^2]}=\frac{D_c \gamma_k M}{D_c M \sigma^2}=\frac{\gamma_k}{\sigma^2}, \ \ k \in \clA_c.\label{snr_def}
\end{align}
We call the vector $\gammam$ or equivalently the diagonal matrix $\Gammam=\diag(\gammam)$ the activity pattern of the users in $\clK_c$. 
We always assume that at each AD slot only a subset $\clA_c \subseteq \clK_c$ of the users of size $A_c:=|\clA_c|$ are active, thus, $\gammam$ is a positive sparse vector with only $A_c$ nonzero elements.
The goal of AD is to identify the subset of all active users $\clA_c$ or a subset thereof consisting of users with sufficiently strong channels $\clA_c(\nu):=\{k \in \clK_c: \gamma_k>\nu \sigma^2\}$, for a pre-specified  threshold $\nu>0$, from the noisy observations as in \eqref{pilot_sig}. 


%
%

Since we  assume that  the channel vectors are spatially
white and
Gaussian, the columns of $\bfY$ in \eqref{pilot_sig} are i.i.d. Gaussian vectors with $\bfY_{:,i} \sim \cg(0, \Sigmam_\bfy)$ where 
\begin{align}\label{eq:true_cov}
\Sigmam_\bfy=\bfA \Gammam \bfA^\herm + \sigma^2 \bfI_M=\sum_{k=1}^{K_c} \gamma_k \bfa_k \bfa_k^\herm + \sigma^2 \bfI_M
\end{align}
the covariance matrix, which is common among all the columns $\bfY_{:,i}$, $i \in [M]$.
We also  define the empirical/sample covariance  of  the columns of the observation  $\bfY$ in \eqref{pilot_sig} as
\begin{align}\label{eq:samp_cov}
\widehat{\Sigmam}_\bfy=\frac{1}{M} \bfY \bfY^\herm =\frac{1}{M} \sum_{i=1}^M \bfY_{:,i} \bfY_{:,i}^\herm.
\end{align}

\subsection{Generalized Signal Model for Activity Detection} 
Before we proceed, it is worthwhile to mention that the signal model for AD in \eqref{sig_eq1}  can be generalized in several directions:
\vspace{1mm}

\noindent { \bf 1})\,\textit{Only a fraction of signal dimensions in a CB is devoted to AD while the remaining signal dimensions are kept for communication:} This reduces pilot dimension $D_c$ in \eqref{sig_eq1}, thus, the length/dimension of the pilots sequences $\{\bfa_k: k \in \clK_c\}$ assigned to the users but preserves the number of antennas $M$. This is beneficial  when the dimension of the CB $D_c$ is significantly large and the number of active users $\clA_c \subset \clK_c$ is not so large.

\vspace{1mm}

\noindent { \bf 2})\,\textit{More than one, say $\kappa>1$, CBs are devoted to AD of a specific set $\clK_c$ of users:} As the simplest scheme, one can  assume that the same length-$D_c$ pilot sequence $\bfa_k \in \bC^{D_c}$ of each user $k \in \clK_c$ is just repeated across the signal dimensions of all $\kappa$ CBs. Due to spatially white and  Gaussian assumption of the channel vectors $\bfh_k$ in \eqref{sig_eq1}, this has the same effect as having only a single CB consisting of $D_c$ signal dimensions while effectively increasing the number of antennas to $M'=\kappa M$, i.e., by a factor of $\kappa$. In general, instead of repeating the  pilot sequence of each user over $\kappa$ CBs, one can vary the pilot sequence of each user over different CBs. This yields more well-conditioned pilot sequences (sufficient randomness and better averaging over different CBs) but does not change the underlying scaling law asymptotically for large  $D_c$, namely, this is still equivalent to having the same pilot dimension $D_c$ but increasing the number of antennas to $M'=\kappa M$.

For simplicity, in the rest of the paper, we always assume that the AD is done over an individual CB by using the whole $D_c$ signal dimensions in the CB.

\section{Proposed Algorithms for Activity Detection}
We will propose now three different estimators to solve the activity
  problem with different assumptions on the underlying statistics of the channel vectors and the sparsity of the activity pattern of the users $\gammam$. 

\subsection{Maximum Likelihood Estimate and Identifiability Condition}
We first consider the Maximum Likelihood (ML) estimator of $\gammam$ by making explicit use of Gaussianity, where after  normalization and simplification we have
\begin{align}
f(\gammam)&:=-\frac{1}{M}\log p(\bfY| \gammam)\stackrel{(a)}{=}-\frac{1}{M}\sum_{i=1}^M \log p(\bfY_{i,:}| \gammam)\\
&=  \log | \bfA \Gammam \bfA^\herm+ \sigma^2 \bfI_{M}| \nonumber \\
&+ \tr\left ( \Big( \bfA \Gammam \bfA^\herm+ \sigma^2 \bfI_{D_c}\Big ) ^{-1} \widehat{\Sigmam}_\bfy \right),\label{eq:ML_cost}
\end{align} 
where $(a)$ follows from the fact that the columns of $\bfY$ are i.i.d. (due to the spatially white user channel vectors), and where $\widehat{\Sigmam}_\bfy$ denotes the sample covariance matrix of the columns of $\bfY$ as in \eqref{eq:samp_cov}.
Note that for spatially white channel vectors considered here, $\widehat{\Sigmam}_\bfy \to \Sigmam_\bfy$ as the number of antennas $M \to \infty$. 
%
%
We also have the following  result.
\begin{theorem}\label{FN_thm}
	Consider the signal model \eqref{pilot_sig} for AD. Then, the empirical covariance matrix $\widehat{\Sigmam}_\bfy$  is a sufficient statistics for the activity pattern of the users $\gammam$. \hfill $\square$
\end{theorem}
\begin{proof}
From \eqref{eq:ML_cost}, it is seen that the log-likelihood of $\bfY$ depends on $\bfY$ only through the sample covariance matrix $\widehat{\Sigmam}_\bfy$. Thus, from the Fischer-Neyman factorization theorem 
\cite{fisher1922mathematical, neyman1936teorema},
$\widehat{\Sigmam}_\bfy$ is a sufficient statistics for estimating $\gammam$. 
\end{proof}
\begin{remark}\label{rem:correlation}
Although we derived \eqref{eq:ML_cost} for spatially white channel vectors, the sample covariance matrix $\widehat{\Sigmam}_{\bfy}$ still provides an almost sufficient statistics for estimation of $\gammam$ as far as the channel vectors are not highly correlated as happens, for example, in \textit{line-of-sight} (LoS) propagation scenarios. Overall, correlation among the components of a user channel vector can be  seen as  reducing the effective number of antennas $M$ (extreme case of $M=1$ antenna in the LoS scenario), which degrades the  AD performance accordingly. \hfill $\lozenge$
\end{remark}

We define the  \textit{maximum likelihood} (ML) estimate of $\gammam$ as
\begin{align}
\gammam^* = \argmin_{\gammam \in \bR_+^{K_c}} f(\gammam).\label{a_ML}
\end{align}
Note that due to the invariance of the ML estimator with respect to transformation of the parameters \cite{casella2002statistical}, $\widehat{\Sigmam}_\bfy$ is also a sufficient statistics for identifying the set of active users $\widehat{\clA}_c(\nu):=\{i: {\gamma}^*_i >\nu\sigma^2\}$, where  $\nu>0$ is a suitable threshold.
%
%
Before we proceed, it is worthwhile to investigate conditions under which the activity pattern $\gammam$ is identifiable. 
We need some notation first. We define the convex cone produced by the rank-1 \textit{positive semi-definite} (PSD) matrices $\{\bfa_k\bfa_k^\herm: k \in \clK_c\}$ generated by  pilot sequences as
\begin{align}
\cone(\bfa_{\clK_c})=\left\{\sum_{k \in \clK_c} \beta_k \bfa_k \bfa_k^\herm: \beta_k \in \bR_+\right\}.
\end{align} 
We also use a similar notation $\cone(\bfa_\clV)$ for the cone generated the pilot sequences of a subset  of users $\clV \subseteq \clK_c$. Note that for any $\clV$, the  cone $\cone(\bfa_\clV)$ is a sub-cone of the cone of PSD  matrices. 
With this notation, we can also write the ML estimation in \eqref{a_ML} equivalently as estimating the cone corresponding to the set of active users.
 To make sure that the set of active users are identifiable, we need to impose the following identifiability condition on the pilot sequences:
\begin{align}
\cone(\bfa_\clV) \not = &\cone(\bfa_{\clV'}), \text{ for all } \clV, \clV' \subseteq \clK_c\nonumber\\
& \text{ where } \clV \not =\clV', \text{ with }  |\clV|, |\clV'| \leq \vartheta, \label{ident_cond}
\end{align}
where $\vartheta$ is set to a sufficiently small number to make sure that condition \eqref{ident_cond} is fulfilled.
%
Note that since the space of $D_c\times D_c$ PSD matrices, denoted by $\clS_{D_c}^+$, has (affine) dimension $\frac{D_c(D_c+1)}{2}$, the number of active $\vartheta$ users should be less than $\frac{D_c(D_c+1)}{2}$ for \eqref{ident_cond} to be fulfilled.
This bound can be also seen to be tight from  Carath\'eodory theorem \cite{caratheodory1907variabilitatsbereich}.
%
%
However, interestingly, this does not restrict  the number of potential users $K_c=|\clK_c|$. For example, when the pilots sequences $\bfa_k$, $k \in \clK_c$, are sampled i.i.d. from an arbitrary continuous distribution, no matter how large $K_c$ is, 
with probability $1$, all the sub-cones in \eqref{ident_cond} would be different and the identifiability condition  would be fulfilled provided that $\vartheta < \frac{D_c(D_c+1)}{2}$ (as also stated in Theorem 1 in \cite{docomo}). This implies that, at least in theory, one can identify the activity of $A_c=\frac{D_c(D_c+1)}{2}$ users in a set of users of arbitrary size $K_c$.  In practice, however, due to the presence of the noise,  one should also limit $K_c$ in order to guarantee a stable AD.

%
Now let us focus on the ML estimation in \eqref{eq:ML_cost}. It is not difficult to check that $f(\gammam)$ in \eqref{eq:ML_cost} is the sum of the concave function $\gammam \mapsto \log | \bfA \Gammam \bfA^\herm+ \sigma^2 \bfI_{D_c}| $ and the convex function $\gammam \mapsto \tr\left ( \Big( \bfA \Gammam \bfA^\herm+ \sigma^2 \bfI_{D_c}\Big ) ^{-1} \widehat{\Sigmam}_\bfy \right)$, so it is not convex in general. However, the following theorem implies that under suitable conditions, the global minimum of $f(\gammam)$ can be calculated exactly.
\def\pscone{\clS_{D_c}^+}
\def\pqcone{\clQ_{D_c}^+}
\begin{theorem}\label{glob_min}
	Let $f(\gammam)$ be as in \eqref{eq:ML_cost}.  Suppose that the pilot sequences $\{\bfa_k: k \in \clK_c\}$ are such that  $\cone(\bfa_{\clK_c})$  coincides  with the cone of  $D_c\times D_c$ PSD matrices $\pscone$.  Then, $f(\gammam)$ has only global minimizers. \hfill $\square$
\end{theorem}
\begin{proof}
Proof in Appendix \ref{glob_min_app}.
\end{proof}

Note that in order for  $\cone(\bfa_{\clK_c})$ to be a good
approximation of the cone of $D_c\times D_c$ PSD matrices as in
Theorem \ref{glob_min}, $K_c$ should be larger than the dimension of
the space of PSD matrices, i.e., $K_c \gtrsim O(D_c^2)$. Interestingly, as we will see later, this holds in the desired scaling regime that we derive for AD problem and is the highly desirable setup for IoT applications.   
Finding the ML estimate $\gammam^*$ in \eqref{a_ML} requires optimization of $f(\gammam)$ over the positive orthant $\bR_+^{K_c}$. In view of Theorem \ref{glob_min}, we can apply simple off-the-shelf algorithms such as gradient descent followed by projection onto $\bR_+^{K_c}$ to find $\gammam^*$. In Appendix \ref{ML_app}, we derive such a coordinate-wise descend algorithm, where at each step we optimize $f(\gammam)$ with respect to only one of its arguments $\gamma_k$, $k \in [K_c]$, and we iterate several times over the whole set of variables until convergence. We can also include the noise variance $\sigma^2$ as an optimization parameter and estimate it along with the parameters $\gamma_k$, $k \in [K_c]$. As $f(\gammam)$ has no local minimizers from Theorem \ref{glob_min} (in the regime we consider in this paper), the coordinate-wise optimization will converge to the unique global minimizer of $f(\gammam)$ with sufficiently many iterations. This coordinate-wise optimization has a closed form expression as derived in Appendix \ref{ML_app}, and is summarized in Algorithm \ref{tab:ML_coord}.

\begin{algorithm}[t]
	\caption{Activity Detection via Coordinate-wise Descend } 
	\label{tab:ML_coord} 
{\small
	\begin{algorithmic}[1]
		\State {\bf Input:} The sample covariance matrix $\widehat{\Sigmam}_\bfy=\frac{1}{M} \bfY \bfY^\herm$ of the $D_c \times M$ matrix of samples $\bfY$.
		\State {\bf Initialize:} $\Sigmam=\sigma^2 \bfI_{D_c}$, $\gammam={\bf 0}$.
		\For { $i=1,2, \dots$}
		\State {Select an index $k \in [K_c]$ corresponding to the $k$-th component of  $\gammam=(\gamma_1, \dots, \gamma_{K_c})^\transp$  randomly or according to a specific schedule.}
		\vspace{2mm}
		\State {\bf ML:} Set $d^*= \max \Big \{\frac{ \bfa_k^\herm \Sigmam^{-1} \widehat{\Sigmam}_\bfy \Sigmam^{-1} \bfa_k -  \bfa_k^\herm \Sigmam^{-1}\bfa_k }{(\bfa_k^\herm \Sigmam^{-1}\bfa_k )^2}, -\gamma_{k} \Big \}$
		\State {\bf MMV:} Set $d^*=\max \Big \{  \frac{\sqrt{\bfa_{k}^\herm \Sigmam^{-1} \widehat{\Sigmam}_\bfy \Sigmam^{-1} \bfa_{k}} -1}{\bfa_{k}^\herm \Sigmam^{-1} \bfa_{k} }, -\gamma_{k} \Big \}$
		
		\State {\bf NNLS:} Set $d^*=\max \Big \{  \frac{\bfa_k^\herm (\widehat{\Sigmam}_\bfy - \Sigmam ) \bfa_k}{\|\bfa_k\|^4}, -\gamma_{k} \}$
		
		\State Update $\gamma_{k} \leftarrow \gamma_{k}+ d^*$.
		\State Update $\Sigmam \leftarrow \Sigmam+ d^* (\bfa_{k} \bfa_{k}^\herm)$.
		\EndFor
		\State {\bf Output:}  The resulting estimate $\gammam$.
	\end{algorithmic}}
\end{algorithm}

\subsection{Multiple Measurement Vector Approach}
Let us consider the signal model $\bfY=\bfA \Gammam^{\frac{1}{2}} \bfH + \bfZ$ in \eqref{pilot_sig} and let us define $\bfX:=\Gammam^\frac{1}{2} \bfH$. Then, we can write  \eqref{pilot_sig} as 
\begin{align}\label{pilot_sig3}
\bfY=\bfA \bfX + \bfZ.
\end{align}
This signal model also arises in a
compressed sensing  setup \cite{donoho2006compressed, candes2006near}, where one tries to recover a structured signal  $\bfX$ from the noisy measurements $\bfY$, where the $D_c \times K_c$ pilot matrix $\bfA=[\bfa_1, \dots, \bfa_{K_c}]$ plays the role of  the measurement matrix. Note that since the activity patterns $\gammam$ is a sparse vector, it induces a common (joint) sparsity among the columns of $\bfX:=\Gammam^\frac{1}{2} \bfH$. Recovery of $\bfX$ from measurements $\bfY$ in \eqref{pilot_sig3} is typically known as the \textit{Multiple Measurement Vector} (MMV) problem in compressed sensing  \cite{tropp2006algorithms, lee2012subspace}. 
  A quite well-known technique for recovering the row-sparse matrix $\bfX$ in the MMV setting is $\ell_{2,1}$-norm Least Squares ($\ell_{2,1}$-LS) as in 
\begin{align}\label{l21_ls}
{\bfX^*}=\argmin_{\bfX} \frac{1}{2} \|\bfA \bfX - \bfY\|_{\sfF}^2 + \varrho \sqrt{M} \|\bfX\|_{2,1},
\end{align}
where $\varrho>0$ is a regularization parameter and where $\|\bfX\|_{2,1}=\sum_{i=1}^{K_c} \|\bfX_{i,:}\|$ denotes the $\ell_{2,1}$-norm of the matrix $\bfX$ given by the sum of $\ell_2$-norm of its $K_c$ rows. It is well-known that $\ell_{2,1}$-norm regularization in \eqref{l21_ls} promotes the sparsity of the rows of the resulting estimate ${\bfX^*}$ and seems to be a good regularizer for detecting the user activity pattern $\gammam$ in our setup. 
Intuitively speaking, we expect that we should obtain an estimate of activity pattern of the users $\gammam$, i.e., the strength of the their channels, by looking at the $\ell_2$-norm of the rows of the estimate ${\bfX^*}$. We have the following result.
\begin{theorem}\label{mmv_thm}
	Let $\bfX^*$ be the optimal solution of $\ell_{2,1}$-LS as in
        \eqref{l21_ls}.  Let $\gammam^*=(\gamma^*_1, \dots,
        \gamma^*_{K_c}) \in \bR_+^{K_c}$, where
        $\gamma^*_i=\frac{\|\bfX^*_{i,:}\|}{\sqrt{M}}$. Then, $\gammam^*$ is the optimal solution of the  convex optimization problem $\gammam^*=\argmin_{\gamma \in \bD_+}   g(\gammam)$, where $g$ is defined by
	\begin{align}\label{g_mmv_cost}
	g(\gammam):=   \tr(\Gammam) + \tr\Big ( ( \bfA\Gammam \bfA^\herm + \varrho \bfI_{D_c})^{-1} \widehat{\Sigmam}_{\bfy}  \Big ),
	\end{align}
	where $\Gammam=\diag(\gammam)$ and $\widehat{\Sigmam}_{\bfy}$ is the sample covariance  \eqref{eq:samp_cov}. \hfill $\square$
\end{theorem}
\begin{proof}
The proof follows from Theorem 1 in \cite{steffens2017compact} and is
included in Appendix \ref{app_mmv} for the sake of
completeness.

\end{proof}

We will prove that in terms of AD we can aim for a scaling regime $A_c=O(D_c^2)$,  where indeed $A_c\gg D_c$. However, when $A_c \gg D_c$, the number of active (nonzero) rows in $\bfX$ is much larger than the number of available measurements $D_c$ via the matrix $\bfA$ in \eqref{l21_ls}. In this regime, any algorithm (including $\ell_{2,1}$-LS) would have a considerable distortion while decoding $\bfX$. Even a \textit{Minimum Mean Squared Error} (MMSE) estimator that knows the exact location of active rows (the support of $\gammam$) will have a significant MSE distortion while decoding the active rows (although it can perfectly identify the inactive ones from the knowledge of the support of $\gammam$). 
It is interesting and highly non-trivial that when only estimation of the activity pattern $\gammam$ is concerned, due to the implicit underlying averaging effect of $\ell_{2,1}$-LS stated in Theorem \eqref{mmv_thm}, namely, the fact that $\gamma^*_i=\frac{\|\bfX^*_{i,:}\|}{\sqrt{M}}$ is given by $\ell_2$-norm of the $i$-th row  of $\bfX^*$, even a noisy estimate $\bfX^*$ is sufficient to yield a reliable estimate of $\gammam$.
Theorem \ref{mmv_thm} also implies that  as far as the strengths (i.e., $\ell_2$-norm) of the rows of $\bfX^*$ are concerned, $\ell_{2,1}$-LS in \eqref{l21_ls} can be simplified to the convex optimization  \eqref{g_mmv_cost}, which depends on the observations $\bfY$ through their empirical covariance, which is a sufficient statistics for identifying $\gammam$ from Theorem \ref{FN_thm}.

Note that the function $g(\gammam)$ in \eqref{g_mmv_cost} is a convex function, thus, finding the optimal solution $\gammam^*$ can be posed as a convex optimization problem over the positive orthant $\gammam \in \bR_+^{K_c}$, which can be solved by conventional convex optimization techniques. Using the well-known Schur complement condition for positive semi-definiteness (see \cite{boyd1994linear} page 28), one can also pose  the optimization in \eqref{g_mmv_cost} as the following \textit{semi-definite program} (SDP):
	\begin{align}\label{sdp_eq}
	(\gammam^*, \Deltam^*)=&\argmin \tr(\Gammam) + \tr(\Lambdam)\nonumber\\
	&\text{ s.t. }  \left (\begin{matrix} \bfA \Gammam \bfA^\herm + \varrho \bfI_{D_c} & \Deltam \\ \Deltam^\herm &\Lambdam \end{matrix}\right ) \succeq {\bf 0},
	\end{align}
	where $\Lambdam$ denotes the auxiliary $D_c\times D_c$ optimization variable, where $\Gammam=\diag(\gammam)$ and where $\Deltam$ is the square root of $\widehat{\Sigmam}_{\bfy}$ (i.e., $\widehat{\Sigmam}_{\bfy}=\Deltam \Deltam^\herm$). In contrast with the MMV optimization in \eqref{l21_ls}, the complexity of \eqref{sdp_eq} does not grow by increasing the number of antennas $M$. 
	\begin{remark}
		In \eqref{sdp_eq} one can treat also the variable $\varrho$ as a free optimization variable to directly estimate also the noise power (note that the resulting cost function is still a jointly convex function of both variables $(\gammam, \varrho)$ and has a unique global optimal solution). As a result, the algorithm does not require the knowledge of the noise power at the BS (although, in practice, the noise power can be easily estimated at the BS from the received noise at un-allocated  signal dimensions). \hfill $\lozenge$
	\end{remark}

	 In this paper, we propose  a simple  coordinate-wise descend algorithm for minimizing the MMV cost function \eqref{g_mmv_cost}. The derivation is very similar to that of the coordinate-wise descend algorithm for the ML cost function in Appendix \ref{ML_app} and has the closed-form expression summarized in Algorithm \ref{tab:ML_coord}.

\subsection{Non-Negative Least Squares}
The last algorithm that we introduce and also theoretically analyze here is based on Non-Negative
Least Squares (NNLS) proposed in \cite{docomo}. We consider the signal model \eqref{pilot_sig} and
$\widehat{\Sigmam}_\bfy$ as in \eqref{eq:samp_cov}. In the NNLS, we
match $\widehat{\Sigmam}_\bfy$ to the true $\Sigmam_\bfy$ in
\eqref{eq:true_cov} in $\ell_2$--distance:
\begin{align}\label{eq_nnls}
\gammam^*=\argmin_{\gammam \in \bR_+^K} \|\widehat{\Sigmam}_\bfy - \bfA \diag(\gammam) \bfA^\herm -\sigma^2 \bfI_{D_c}\|^2_{\sfF}.
\end{align}
  Let $\widehat{\sigmam}_\bfy=\vec(\widehat{\Sigmam}_\bfy)$ denotes the $D_c^2 \times 1$ vector obtained by stacking the columns of $\widehat{\Sigmam}_\bfy$ and let $\bA$ be an $D_c^2 \times K_c$ matrix whose $k$-th column, $k \in [K_c]$, is given by $\vec(\bfa_k \bfa_k^\herm)$. Then, we can write \eqref{eq_nnls} in the convenient form
 \begin{align}\label{eq_nnls_vec}
   \gammam^*=\argmin_{\gammam \in \bR_+^K} \|\widehat{\sigmam}_\bfy - \bA \gammam -\sigma^2 \vec(\bfI_{D_c})\|^2_2,
 \end{align}
 as a \textit{least squares} problem with non-negativity constraint,
 known as \textit{nonnegative least squares} (NNLS).  
If
 the row span of $\bA$ intersects the positive orthant, NNLS
 implicitly also performs $\ell_1$-regularization, as discussed for
 example in \cite{slawski2013non} and called 
 $\mathcal{M}^+$-criterion in \cite{kueng2016robust}.  Because of
 these features, NNLS has recently gained interest in many applications in signal
 processing \cite{song2017scalable}, 
  compressed
 sensing \cite{kueng2016robust}, and machine learning.  In our case the
 $\mathcal{M}^+$--criterion is fulfilled in an optimally--conditioned
 manner and allows us to establish the following result:
%

\begin{theorem}
  Let $\{\avec_k\}_{k=1}^\dimParam\subset \mathbb{C}^{\dimPilots}$
  independent copies of a random vector with iid. unit-magnitude
  entries and $\dimParam\geq 16$. Fix some
  $\delta\in[8/\dimParam,4/\sqrt{41})$. If
  \begin{equation}
    \dimPilots(\dimPilots-1)\geq c' \delta^{-2}s\log^2(e\dimParam/s)
  \end{equation}
   then with overwhelming probability the following holds: For all activity pattern vectors 
  $\gammam^\circ$, the solution $\bxnnls$ of \eqref{eq_nnls_vec} is
  guaranteed to fulfill for $1\leq p\leq 2$:
  \begin{equation}
    \begin{split}
      \lVert \bxtrue - \bxnnls\rVert_{\ellp{p}} 
      \leq
      &\frac{2C}{s^{1-\frac{1}{p}}} \sigma_s (\bxtrue)_{\ellp{1}} +
      \frac{2D}{s^{\frac{1}{2}-\frac{1}{p}}} \left( 1+
        \frac{\lambda\tau'\sqrt{\dimPilots}}{\sqrt{\dimPilots-1}} \right)
      \frac{\lVert \vc{d}\rVert_{\ellp{2}}}{\dimPilots}\\
    \end{split}
    \label{eq:short:thm:mse}
  \end{equation}
  where  $\sigma_s (\bxtrue)_1$ denotes the $\ell_1$--norm of
  $\gammam^\circ$ after removing its $s$ largest components, where 
   $\bfd=\vec(\widehat{\Sigmam}_\bfy - \sum_{k=1}^{K_c} \gamma^\circ_k
  \bfa_k \bfa_k^\herm)$, and where $C$, $D$, and $\tau'$ only depend
  on $\delta$, and $\lambda$ is a global constant.  
  \label{thm:nnls:mse}
\end{theorem}

The proof, given in the appendix \ref{appendix:nnlsproof}, is based on
a combination of the NNLS results of \cite{kueng2016robust} and an
extension of RIP-results for the heavy-tailed column-independent model
\cite{Adamczak2011,Guedon2014:heavy:columns}.  We like to mention that
the theorem holds even for more general random models and the constant
can be computed explicitly.  For example, $\delta=0.5$ gives
$\tau'\approx 1.5$, $C\approx 8.6$ and $D\approx 11.3$.
The result is uniform meaning that with high probability (on a draw of
$\bA$) it holds {\em for all} $\gammam^\circ$. For
$s=A_c=\lVert\gammam^\circ\rVert_0$ it implies (up to
$\|\vc{d}\|_{2}$-term) exact recovery since then
$\sigma_s (\bxtrue)_{\ellp{1}}=0$.  A relevant extension of this
result to the case $p\rightarrow\infty$ whould be important but, in
this generality, it is known that then one can not hope for a linear
scaling in $s$ (see here for example \cite[Theorem
3.2]{dirksen:ripgap}).  Since
$\|\cdot\|_{\ellp{\infty}}\leq\|\cdot\|_{\ellp{p}}$ our result
\eqref{eq:short:thm:mse} also implies an estimate for the
communication relevant $\ell_\infty$-case but with sub-optimal scaling
(we will discuss this below).  Furthermore improvements for this
particular case may be possible in the non-uniform or averaged case,
as it has been investigate for the subgaussian case in
\cite{slawski2013non}.

A straightforward analysis in Appendix \ref{cov_err_app} shows that for Gaussian and spatially white
user channel vectors, the statistical fluctuation term $\bfd$ in
Theorem \ref{thm:nnls:mse} concentrates very well around its mean given by
\begin{align}
\|\bfd\|&=\|\vec(\widehat{\Sigmam}_\bfy - \sum_{k=1}^{K_c} \gamma^\circ_k \bfa_k \bfa_k^\herm - \sigma^2 \bfI_{D_c})\|= \|\vec(\widehat{\Sigmam}_\bfy - \Sigmam_\bfy)\| \nonumber\\
& =\|\widehat{\Sigmam}_\bfy - \Sigmam_\bfy\|_\sfF\nonumber\\
&\approx \sqrt{\bE[\|\widehat{\Sigmam}_\bfy - \Sigmam_\bfy\|_\sfF^2]}\nonumber\\
&= \frac{\tr(\Sigmam_\bfy)}{\sqrt{M}}.
\end{align}
 Assuming that at a specific AD slot only $A_c$
of the users are active and setting $s=A_c$ in Theorem \ref{thm:nnls:mse},
using $\sigma_s(\gammam^\circ)=0$ for $s=A_c$, and the well-known
inequality $\|\gammam^\circ\|_1 \leq \sqrt{A_c} \|\gammam^{\circ}\|_2$,
we obtain the following scaling law on the performance of NNLS
\begin{align}
\lVert \bxtrue - \bxnnls\rVert_{\ellp{2}} \leq
c_3 \left ( \frac{\sigma^2}{\sqrt{M}} + \sqrt{\frac{A_c}{M}} \|\gammam^\circ\|_2 \right  )\label{gamma_perf}
\end{align}
with a constant $c_3=2c_2 \big( 1 + \frac{\tau'\sqrt{\dimPilots}}{\sqrt{\dimPilots-1}}\big)=O(1)$ provided that 
\begin{align}
A_c \log^2(\frac{K_c}{ A_c}) \lesssim O(D_c^2). \label{A_c_perf}
\end{align}
It is important to note that: {\bf i})\,from \eqref{gamma_perf} the error term only depends on the number of active users $A_c$ and vanishes if the number of antennas $M$ is slightly larger than $A_c$ (i.e., $A_c/M=o(1)$) so that sufficient statistical averaging happens, {\bf ii})\,from \eqref{A_c_perf} NNLS can identify as large as $O(D_c^2)$ active users by paying only a poly-logarithmic penalty  $O(\log^2 (\frac{K_c}{A_c}))$ for increasing the number of potential users $K_c$. As stated before, this is a very appealing scaling law  in setups such as IoT and D2D, where $K_c$ can be dramatically large.

\section{Simulation Results}
In this section, we evaluate the performance of our algorithms via numerical simulations. 
\subsection{Performance Metric}
We assume that the output of each algorithm is an estimate $\gammam^*$ of the activity pattern of the users. 
We define $\widehat{\clA}_c(\nu):=\{i: {\gamma}^*_i >\nu\sigma^2\}$, with  $\nu>0$, as the estimate of the set of active users.
We also define the  detection/false-alarm probability averaged over the active/inactive users as
\begin{align}
\probd(\nu) =\frac{\bE[| \clA_c \cap \widehat{\clA}_c|]}{A_c},\ \  \probf(\nu) =\frac{\bE[| \widehat{\clA}_c \backslash \clA_c|]}{K_c-A_c}
\end{align}
where $A_c$ and $K_c$ denote the number of active and the  number of potential users, respectively. By varying $\nu \in \bR_+$, we plot the \textit{Receiver Operating Characteristic} (ROC) \cite{poor2013introduction} as a measure of performance of our  proposed algorithms.

\subsection{Comparison with the Literature}
\subsubsection{Vector Approximate Message Passing}
We compare the performance of our proposed algorithm with that of \text{Vector Approximate Message Passing} (VAMP) algorithm proposed for AD in \cite{liu2017massive, chen2018sparse}. Given the noisy observation $\bfY=\bfA \bfX + \bfZ$ as in \eqref{pilot_sig3}, VAMP recovers an estimate of the row-sparse signal $\bfX$ as follows
\begin{align}
\bfX^{t+1}&=\eta(\bfA^\herm \bfR^t + \bfX^t, \bfg, t),\label{vamp1}\\
\bfR^{t+1}&=\bfY - \bfA \bfX^{t+1} + \frac{K_c}{D_c} \bfR^t \langle \eta'(\bfA^\herm \bfR^t + \bfX^t, \bfg, t) \rangle \label{vamp2}
\end{align}
where $\eta$ is the MMSE vector denoising function applied row-wise,  where $\eta'$ denotes the component-wise derivative of $\eta$, where $\bfg=(g_1, \dots, g_{K_c})^\transp$ denotes the large-scale fading coefficients of the users, where $\langle. \rangle$ denotes the averaging operator, and $\bfR^t \langle \eta'(\bfA^\herm \bfR^t + \bfX^t, \bfg, t)$ is the well-known Onsager's term. It was shown in \cite{kim2011belief, ziniel2013efficient} that, due to the Onsager correction, in the asymptotic regime of $D_c, A_c,K_c \to \infty$, the noisy input $\widetilde{\bfX}^t:=\bfA^\herm \bfR^t + \bfX^t$ to the MMSE vector denoiser decouples to (for spatially white user channel vectors)
\begin{align}\label{decoupled_model}
\widetilde{\bfX}^t_{i,:} = \bfX^t_{i,:} + \tau^t \bfN_{i,:}, \ i\in [K_c],
\end{align}
where $\bfN$ is an $K_c \times M$ matrix consisting of i.i.d. $\cg(0,1)$ elements, and where $\tau^t$ is obtained by a simple state evolution (SE) equation $\tau^{t}=\xi(\tau^{t-1}, \bfg)$. We refer to \cite{liu2017massive, chen2018sparse} for further explanation of VAMP and the derivation of MMSE denoiser and the SE equation. 
For simulations, we consider an ideal scenario for VAMP where the large-scale fading coefficients $\bfg$ are known at the BS and are used in the VAMP algorithm and also in the derivation of the SE equation. 

\subsubsection{Empirical Genie-aided tuning of VAMP (g-VAMP)}
We have observed empirically that, for large number of BS antennas $M$, VAMP is quite sensitive and numerically unstable when the number of active users $A_c$ is much larger than the pilot dimension $D_c$. To avoid this numerical instability, we apply a genie-aided tuning of the VAMP, where at each iteration we use the rows of  $\widetilde{\bfX}^t:=\bfA^\herm \bfR^t + \bfX^t$ (the noisy input to the MMSE vector denoiser as in \eqref{vamp1}) corresponding to the inactive users to estimate the  noise variance $\tau^t$ in the decoupled model \eqref{decoupled_model}.
We call the resulting algorithm g-VAMP and use it for comparison.

\subsubsection{Comparison with our proposed algorithms}
Our proposed algorithms have the following advantages compared with VAMP:
\begin{enumerate}
	\item they do not require the knowledge of large-scale fading coefficients of the channel $\bfg=(g_1, \dots, g_{K_c})$ (or its statistics) and the number of active users $A_c$, which is required for deriving (and tuning) the VAMP algorithm but is difficult to have in IoT setups due to the sporadic nature of the  traffic of the users (especially $A_c$).
	
		\item they require only the sample covariance of channel observations, thus, they are quite robust to variation in the statistics (i.e., Gaussianity, spatial correlation, etc.) of the user channel vectors, whereas having the  knowledge of the statistics of the channel vectors is quite critical for deriving the MMSE vector denoiser in VAMP and  obtaining the appropriate SE equation.
		
		\item they do not require any  tuning and are numerically stable.

	
\end{enumerate}

\begin{figure}[t]
	\centering
	\includegraphics[scale=0.75]{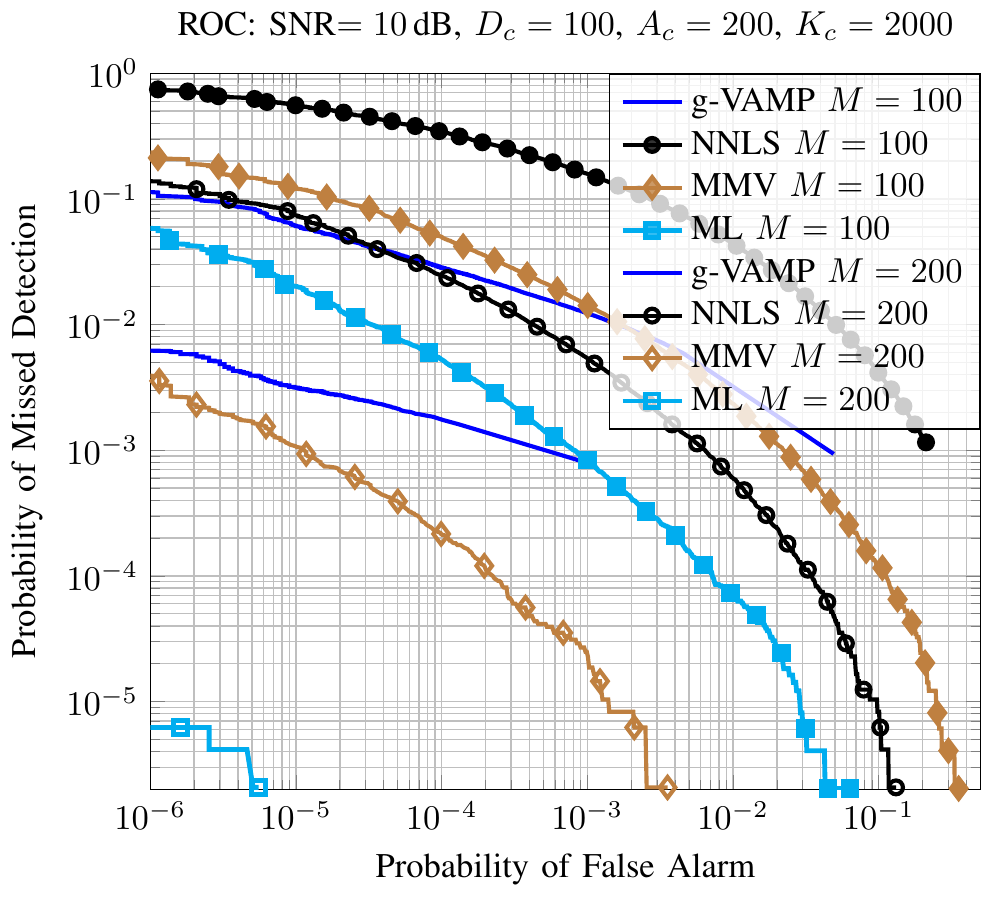}
	\caption{Scaling of the performance of algorithms vs. $M$.}
	\label{fig:scale_with_M}
\end{figure}%
\subsection{Scaling with Number of Antennas $M$}
In this section, we compare the performance of our proposed algorithm for different number of BS antennas $M$. We consider a CB of dimension $D_c=100$, and a total number of $K_c=2000$ users assigned to the CB. We assume that at each slot only $A_c=200$ out of $K_c=2000$ users are active. We also assume a symmetric scenario where all the active users have the same large-scale fading coefficients $g_k, k\in [K_c]$, and an SNR of $\frac{g_k}{\sigma^2}=10$\,dB. It is worthwhile to note that our algorithms do not require the knowledge of  channel strengths $g_k, k\in [K_c]$. Fig.\,\ref{fig:scale_with_M} illustrated the simulation results. It is seen that 
\begin{enumerate}
\item NNLS does not perform well when the number of antennas is less the the number of active users ($M=100$, $A_c=200$)  but its performance dramatically improves by increasing the number of antennas ($M=A_c=200$).

\item As reported in \cite{liu2017massive, chen2018sparse}, VAMP performs very well when the pilot dimensions is close to the number of active users ($D_c \lesssim A_c$). Our results also confirm this and illustrate that in the more interesting regime of $D_c \ll A_c$ (here $D_c=100$, $A_c=200$) suited for IoT, VAMP performance is comparable to that of MMV for $M=100$ while MMV works much better for $M=200$. Also, note that MMV as well as NNLS and ML do not require the knowledge of the large-scale fading coefficient of the channel, thus, are quite robust.

\item ML performs much better than all the other algorithms and requires much less number of antennas $M$. 
\end{enumerate}

\subsection{Scaling with the total Number of Users $K_c$}
In this section, we investigate the dependence of the performance of the algorithms on the total number of users $K_c$. We again assume  a CB of dimension $D_c=100$, and $A_c=200$ active users. We set the number of antennas $M=A_c=200$ to guarantee a good performance for all the algorithm. We vary the total number of users in $K_c=\{2000, 4000\}$. Fig.\,\ref{fig:AD_scaling_with_Kc} illustrates the simulation results. It is seen that the performance of NNLS (and that of the other algorithms) is not sensitive to the number of users $K_c$, as expected from the scaling law (the poly-logarithmic dependence on $K_c$)  claimed in Theorem \ref{thm:nnls:mse}.

\begin{figure}[t]
	\centering
	\includegraphics[scale=0.75]{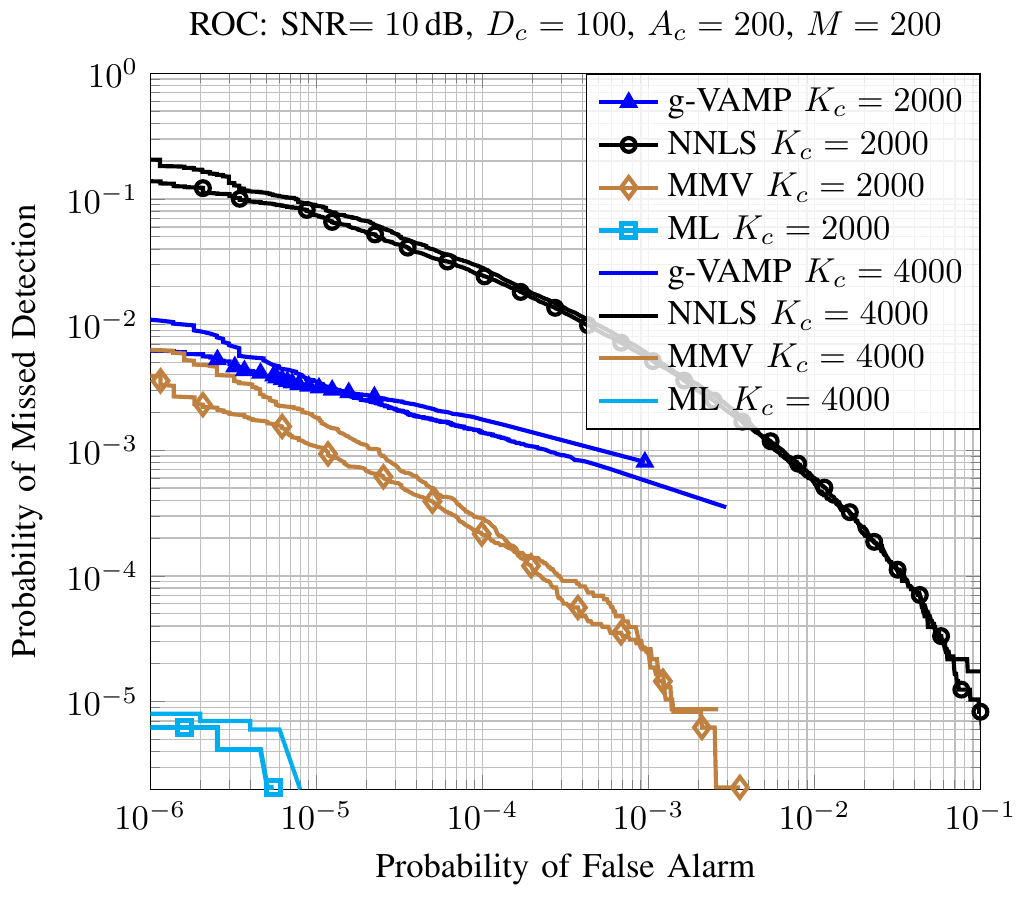}
	\caption{Scaling of performance of the algorithms vs. $K_c$.}
	\label{fig:AD_scaling_with_Kc}
\end{figure}%

\section{Correlated user channel vectors}
In this paper, we always assumed that the channel vector of all the users are spatially white (see also Remark \ref{rem:correlation}). In this section, we repeat the simulations for spatially correlated channel vectors. We consider a simple case where the BS is equipped with a \textit{Uniform Linear Array} (ULA) with $M$ antennas. It is well-known that for the ULA, the spatial covariance matrix of the channel vectors becomes a Toeplitz matrix, which can be diagonalized approximately with the \textit{Discrete Fourier Transform} (DFT) matrix  such that the correlated channel vector $\bfh \in \bC^M$ of a generic user can be written as $\bfh=\bfF_M \bfw$, where $\bfF_M$ denotes the DFT matrix of order $M$ with normalized columns and where $\bfw=(w_1, \dots, w_M)^\transp$ denotes a complex Gaussian vector with independent components $w_i \sim \cg(0, \beta_i)$, where $\sum_{i=1}^M \beta_i=\bE[\|\bfw\|^2]=\bE[\|\bfh\|^2]=M$. In a ULA, the $i$-th column of $\bfF_M$ corresponds to the array response for a planar wave with the \textit{angle of arrival} (AoA) $\theta_i=-\theta_{\max} + \frac{2(i-1) \theta_{\max}}{M}$, $i\in [M]$, where $\theta_{\max}$ denotes the maximum angle covered by the BS antennas, and where $w_i$ denotes the contribution to the channel vector $\bfh$ coming from those scatterers lying in the interval $\theta \in [\theta_i, \theta_{i+1})$. We refer to \cite{haghighatshoar2017massive, haghighatshoar2017low} for  further details.

To create spatially correlated channel vectors, we will assume an angular block-sparse propagation model where 
\begin{align*}
\beta_i= \left \{ 
\begin{array}{ll}  
\frac{M}{M_\text{eff}} & \text{ for }i \in \Big \{i_0,  \dots, (i_0+M_
\text{eff}-1)\,\text{mod}\, M\Big \}, \\
0 & \text{else where,} 
\end{array} 
\right.
\end{align*}
such that only $M_\text{eff}$ consecutive AoAs are active with $i_0$ denoting the index of the start of the block. As also mentioned in Remark \ref{rem:correlation} such a spatial correlation reduces the effective number of antennas from $M$ to $M_\text{eff} <M$. 

For the simulations, we assume that the start of the block $i_0\in[M]$ is selected uniformly at random for each user. We also assume that $\frac{M_\text{eff}}{M}=0.5$ such that the channel vector of each user has the normalized angular spread of $50\%$ (compared with the whole angular spread $2 \theta_{\max}$ of the array).  
Fig.\,\ref{fig:AD_sim_corr} illustrates the simulation results. For the case of correlated channel vectors, we consider $M=400$ antennas and $50\%$ spatial correlation, which amounts to $M_\text{eff}=200$ effective number of antennas. We compare the results with the case of spatially white channel vectors. It is seen from Fig.\,\ref{fig:AD_sim_corr} that the results match very well, which confirms the fact that spatial correlation reduces the effective number of antennas (as stated in Remark \ref{rem:correlation}).

\begin{figure}[t]
	\centering
	\includegraphics[scale=0.75]{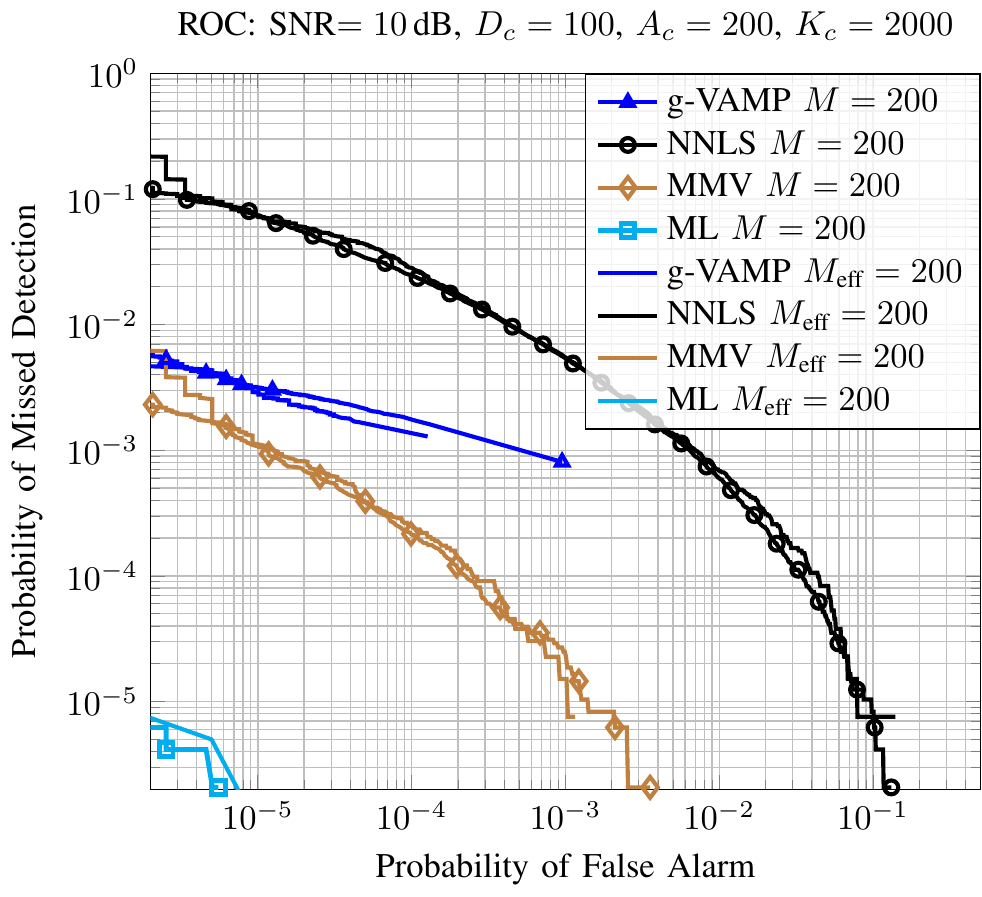}
	\caption{Comparison of the performance of the algorithms for  spatially white and spatially correlated channel vectors with the same effective number of antennas $M=200$ (spatially white) and $M=400$, $M_\text{eff}=200$ (spatially correlated).}
	\label{fig:AD_sim_corr}
\end{figure}%

\section{Conclusion}
 In this paper, we studied the problem of user activity
		detection  in a massive MIMO setup, where the BS
	has $M \gg 1$ antennas. 
	We showed that with 
	a CB containing $D_c$ signal dimensions
	one can stably estimate the activity of $A_c=O(D_c^2/\log^2(\frac{K_c}{A_c}))$ active users in a set of $K_c$ users, which is much larger than  the previous bound $A_c=O(D_c)$
	obtained via traditional compressed sensing techniques. In particular, in our proposed scheme
	one needs
	to pay only a poly-logarithmic penalty
	$O(\log^2(\frac{K_c}{A_c}))$ for increasing the number of potential
	users $K_c$, which makes it ideally suited for activity detection in IoT setups. 
	We proposed  low-complexity algorithms for activity detection and provided numerical simulations to
	illustrate our results. We also compared the performance of our proposed activity detection algorithms with that of other competitive algorithms in the literature.

\appendices
\section{Proofs of the Results}
\subsection{Derivation of the coordinate-wise ML Optimization}\label{ML_app}
In this section, we derive a closed-form expression for the coordinate-wise optimization of the ML cost function $f(\gammam)$ in \eqref{eq:ML_cost}. Let $k \in [K_c]$ be the index of the selected coordinate and let us define 
$f_k(d)=f(\gammam+ d \bfe_k)$ where $\bfe_k$ denotes the $k$-th canonical basis with a single $1$ at its $k$-th coordinate and zero elsewhere.
Setting $\Sigmam=\Sigmam(\gammam)=\bfA \Gammam \bfA^\herm + \sigma^2 \bfI_{D_c}$ where $\Gammam=\diag(\gammam)$ and applying the well-known Sherman-Morrison rank-1 update identity \cite{sherman1950adjustment} we obtain that
\begin{align}
\big (\Sigmam + d \bfa_k \bfa_k^\herm \big )^{-1}=\Sigmam^{-1} - \frac{d\, \Sigmam^{-1} \bfa_k \bfa_k^\herm \Sigmam^{-1}}{1+ d \, \bfa_k^\herm \Sigmam^{-1} \bfa_k}. \label{dumm_ml_app2}
\end{align} 
Using \eqref{dumm_ml_app2} and applying the well-known determinant identity 
\begin{align}
\big | \Sigmam + d \bfa_k \bfa_k^\herm \big | =(1+ d\, \bfa_k^\herm \Sigmam ^{-1} \bfa_k) \big | \Sigmam\big |,
\end{align}
we can simplify $f_k(d)$ as follows
\begin{align}
f_k(d)= c &+ \log (1+ d\, \bfa_k^\herm \Sigmam ^{-1} \bfa_k) \nonumber \\
& -  \frac{  \bfa_k^\herm \Sigmam^{-1} \widehat{\Sigmam}_\bfy \Sigmam^{-1} \bfa_k }{1+ d \, \bfa_k^\herm \Sigmam^{-1} \bfa_k} d \label{dumm_ml_app3}
\end{align}
where $c=\log \big | \Sigmam \big | +\tr( \Sigmam^{-1} \widehat{\Sigmam}_\bfy)$ is a constant term independent of $d$. Note that from \eqref{dumm_ml_app3}, $f_k(d)$ is well-defined only when $d>d_0:=-\frac{1}{\bfa_k^\herm \Sigmam ^{-1} \bfa_k}$. Taking the derivative of $f_k(d)$ yields
\begin{align*}
f_k'(d)= \frac{  \bfa_k^\herm \Sigmam^{-1}\bfa_k }{1+ d \, \bfa_k^\herm \Sigmam^{-1} \bfa_k} - \frac{  \bfa_k^\herm \Sigmam^{-1} \widehat{\Sigmam}_\bfy \Sigmam^{-1} \bfa_k }{(1+ d \, \bfa_k^\herm \Sigmam^{-1} \bfa_k)^2}.
\end{align*}
The only solution of $f_k'(d)$ is given by 
\begin{align}
d^*= \frac{ \bfa_k^\herm \Sigmam^{-1} \widehat{\Sigmam}_\bfy \Sigmam^{-1} \bfa_k -  \bfa_k^\herm \Sigmam^{-1}\bfa_k }{(\bfa_k^\herm \Sigmam^{-1}\bfa_k )^2}.
\end{align}
Note that $d^* \geq d_0=-\frac{1}{\bfa_k^\herm \Sigmam^{-1}\bfa_k }$, thus, one can check from \eqref{dumm_ml_app3} that $f_k$ is indeed well-defined at $d=d^*$.
Moreover, we can check from \eqref{dumm_ml_app3} that $\lim_{\epsilon \to 0^+} f_k(d_0+\epsilon)=\lim_{d \to +\infty} f_k(d)=+\infty$, thus, $d=d^*$ must be the global minimum of $f_k(d)$ in $(d_0, \infty)$. Note that since after the update we have $\gamma_k \leftarrow \gamma_k + d$, to preserve the positivity of $\gamma_k$, the optimal update step $d$ is in fact given by $\max \big \{d^*, -\gamma_k\big \}$ as illustrated in Algorithm \ref{tab:ML_coord}.

\subsection{Proof of Theorem \ref{glob_min}}\label{glob_min_app}
This follows from the \textit{geodesic convexity} of the ML cost function $f(\gammam)$ \cite{wiesel2012geodesic} when $\bfA \Gammam \bfA^\herm$ spans the whole set of $D_c\times D_c$ PSD matrices. Here, we provide another simple proof.

Consider $\cone(\clK_c)=\{\sum_{k=1}^{K_c} \beta_k \bfa_k \bfa_k^\herm: \beta_k\geq 0\}$ as a sub-cone of PSD matrices $\pscone$ produced by the pilot sequences. By our assumption, $\cone(\clK_c)$ approximates the cones of $D_c \times D_c$ PSD matrices $\clS_{D_c}^+$ very well. Let us  first define 
\begin{align}
\clP_{D_c}&:=\{\sigma^2 \bfI_{D_c} + \sum_{k=1}^{K_c} \beta_k \bfa_k \bfa_k^\herm: \beta_k \in \bR_+\}\\
&=\{\sigma^2 \bfI_{D_c} + \bfC: \bfC \in \cone(\clK_c)\}\\
&\stackrel{(a)}{=} \sigma^2 \bfI_{D_c} + \clS_{D_c}^+,
\end{align}
where in $(a)$ we used the assumption that
$\cone(\clK_c)\approx \clS_{D_c}^+$. We also define
$\clQ_{D_c}=\clP_{D_c}^{-1}:=\{\bfP^{-1}: \bfP \in
\clP_{D_c}\}$. It is not difficult to check that
\begin{align}
\clQ_{D_c}=\{\bfQ\in \clS_{D_c}^+: \lambda_{\max}(\bfQ) \leq \frac{1}{\sigma^2}\},
\end{align}
where $\lambda_{\max}$ denotes the largest singular value of a matrix. Since $\lambda_{\max}$ is a convex function over the convex set of PSD matrices $\clS_{D_c}^+$, it results that $\clQ_{D_c}$ is indeed a convex subset of $\clS_{D_c}^+$. Applying this change of variable, we can, therefore, write the ML estimation of $\gammam$ equivalently as the following optimization problem
\begin{align}\label{glob_min_dumm1}
\bfQ^*=\argmin_{\bfQ \in \clQ_{D_c}} - \log |\bfQ| + \tr(\bfQ \widehat{\Sigmam}_{\bfy}).
\end{align} 
Note that since $\bfQ \mapsto - \log |\bfQ| + \tr(\bfQ \widehat{\Sigmam}_{\bfy})$ is a convex function of $\bfQ$ and $\clQ_{D_c}$ is a convex set, \eqref{glob_min_dumm1} is a convex optimization problem, whose local minimizers are global minimizers as well. This implies that the ML cost function $f(\gammam)$ in \eqref{eq:ML_cost} has only global minimizers.

\subsection{Proof of Theorem \ref{mmv_thm}} \label{app_mmv}
The proof follows by extending Theorem 1 in \cite{steffens2017compact}.
The key observation is that for a $\bfu \in \bC^M$, the $\ell_2$ norm $\|\bfu\|$ can be written as the output of the following optimization 
\begin{align}\label{eq:norm_rep}
\|\bfu\|=\min_{\bfv \in \bC^M,\,\delta \in \bC:\, \delta \bfv =\bfu} \frac{1}{2}(\|\bfv\|^2 + |\delta|^2).
\end{align}
In particular, $\|\bfu\|=|\delta^*|^2$, where $s^*$ is the optimal solution of \eqref{eq:norm_rep}. Applying this argument to the rows of $\bfX$, we can write the $l_{2,1}$ norm of $\bfX$ as follows
\begin{align}\label{l21_rep}
\|\bfX\|_{2,1}=\min_{\bfV\in \bC^{K_c\times  M},\, \Deltam\in \bD:\, \Deltam \bfV=\bfX} \frac{1}{2}(\|\bfV\|_{\sfF}^2 + \|\Deltam\|_{\sfF}^2),
\end{align}
where $\bD$ denotes the space of $K_c\times K_c$ diagonal matrices with diagonal elements in $\bC$, and where $\Deltam=\diag(\delta_1, \dots, \delta_{K_c})\in \bD$. In particular, $\|\bfX_{i,.}\|=|\delta_k^*|^2$, where $\Deltam^*=\diag(\delta_1^*, \dots, \delta_{K_c}^*)$ is the optimal solution of \eqref{l21_rep}. Replacing $\|\bfX\|_{2,1}$ with \eqref{l21_rep}, we can transform the optimization problem \eqref{l21_ls}  into 
\begin{align}\label{trans_eq}
(\bfV^*, \Deltam^*)=\argmin_{\bfV\in \bC^{K_c\times  M},\,\Deltam\in \bD} &\frac{1}{\varrho \sqrt{M}} \sum_{k=1}^M \|\bfA \Deltam \bfV_{:,t} - \bfY_{:,k}\|_{\sfF}^2 \nonumber\\
&+ \|\bfV\|_{\sfF}^2 + \|\Deltam\|_{\sfF}^2.
\end{align}
For a fixed $\Deltam$, the minimizing $\bfV$ as a function of $\Deltam$ can be obtained via a least-square minimization, where after replacing the solution in \eqref{trans_eq} and applying the \textit{matrix inversion lemma} \cite{hager1989updating} and further simplifications, we obtain the following optimization in terms of $\Deltam$
\begin{align}\label{l21_rep2}
&\Deltam^*=\argmin_{\Deltam \in \bD}  \tr(\frac{\Deltam \Deltam^\herm}{\sqrt{M}})\nonumber\\
&+\frac{1}{M} \sum_{k=1}^M \tr\Big ( (\bfA \frac{\Deltam \Deltam^\herm}{\sqrt{M}} \bfA^\herm + \varrho \bfI_m)^{-1} \bfy(t) \bfy(t)^\herm \Big ).
\end{align}
Note that this optimization can be reparameterized with $\bfP=\frac{\Deltam \Deltam^\herm}{\sqrt{M}}=\diag(\frac{|\delta_1|^2}{\sqrt{M}}, \dots, \frac{|\delta_a|^2}{\sqrt{M}}) \in \bD_+$, where $\bD_+$ denotes the space of all $K_c\times K_c$ diagonal matrices with positive diagonal elements. Thus, we can write:
\begin{align}\label{l21_rep3}
&\bfP^*=\argmin_{\bfP \in \bD_+}   \tr(\bfP)\nonumber\\
&+\frac{1}{M} \sum_{k=1}^M \tr\Big ( ( \bfA\bfP \bfA^\herm + \varrho \bfI_m)^{-1}  \bfy(t) \bfy(t)^\herm  \Big )
\end{align}
Denoting by $\bfP^*=(p_1^*, \dots, p_a^*)$ the optimal solution of \eqref{l21_rep3}, we have the underlying relation $\frac{\|\bfX^*_{i,.}\|}{\sqrt{M}}=\frac{|\delta^*_i|^2}{\sqrt{M}}=p^*_i$. This implies that the optimal solution $\gammam^*$ in the statement of the theorem, is indeed the solution of the following convex optimization
\begin{align*}
\gammam^*&:=\argmin_{\gamma \in \bD_+}   \sum_{i=1}^{K_c} \gamma_i + \tr\Big ( ( \bfA\diag(\gammam) \bfA^\herm + \varrho \bfI_m)^{-1} \widehat{\Sigmam}_{\bfy}  \Big )\\
&=\argmin_{\gamma \in \bD_+} \tr(\Gammam) + \tr\Big ( ( \bfA\Gammam \bfA^\herm + \varrho \bfI_{D_c})^{-1} \widehat{\Sigmam}_{\bfy}  \Big ),
\end{align*}
where we replaced $ \widehat{\Sigmam}_{\bfy}=\frac{1}{M} \bfY \bfY^\herm =\frac{1}{M} \sum_{k=1}^{M} \bfY_{:,k} \bfY_{:,k}^\herm$ and defined $\Gammam=\diag(\gammam)$. 
This completes the proof.

\subsection{Proof of the Recovery Guarantee for NNLS, Theorem \ref{thm:nnls:mse}}
\label{appendix:nnlsproof}

Let $\{\Acol_k:=\avec_k\avec^\herm_k\}_{k=1}^\dimParam$ be the rank-1
matrices generated by the pilots sequences of the users. We consider
the noisy model $\SigmaEmp=\sum_{k=1}^\dimParam\xtrue_k\cdot \Acol_k+\mt{D}$
where $\bxtrue=(\xtrue_1,\dots,\xtrue_\dimParam)\in \bR_+^{K_c}$
denotes the unknown activity pattern of the users to be estimated. We assume that $\bxtrue$ is
potentially $s$--sparse or compressible (well--approximated by sparse
vectors) and $\mt{D}$ is a residual error matrix. 
In particular, the vectorized problem is:
\begin{equation}
  \sigmaEmp=\text{vec}(\SigmaEmp)=:\Amatrix\bxtrue+\vc{d}.
\end{equation}
where $\vc{d}=\text{vec}(\mt{D})$, $\bxnnls\in\mathbb{R}_+^\dimParam$
and $\bA$ is a $D_c^2 \times K_c$ matrix whose $k$-th column is given
by $\vec(\bfa_k\bfa_k^\herm)$.  We are interested in the behavior of
the NNLS solution \eqref{eq_nnls_vec} for a given sparsity parameter
$s$, typically taken as number of active users $A_c$.  To cover also
the more general compressible regime we define the $\ell_1$-error of
its best $s$-sparse approximation to $\bxtrue$ as:
\begin{equation}
  \sigma_s(\bxtrue)_{\ellp{1}}=\min_{\lVert \xsome\rVert_0\leq s}\lVert \bxtrue-\xsome\rVert_{\ellp{1}}
\end{equation}
The motivation of directly using the unregularized NNLS for recovery
comes from the self-regularizing property of matrices $\Amatrix$
having simultaneously the \emph{$\mathcal{M}^+$--criterion} (its ``row
span intersects the positive orthant'') and the \emph{$\ell_2$--robust
  null space property} of order $s$ ($s$-NSP) (details see,
\cite{kueng2016robust}).\\[.3em]

\newcommand{\AmatrixR}{\Amatrix^R} \newcommand{\AcolR}{\Acol^R}
    First of all, we consider also complex-valued pilot sequences
    $\avec_k$ such that the columns $\vec(\avec_k\avec_k^\herm)$ of
    $\Amatrix$ are complex-valued as well.  The following generic
    recovery result for the NNLS has been shown in \cite[Theorem
    1]{kueng2016robust} (it has been formulated for the real setting
    but it holds for complex matrices as well). First, the row span of
    the complex $\dimPilots^2\times \dimParam$--matrix $\Amatrix$
    intersects the positive orthant (called as
    $\mathcal{M}^+$-criterion in \cite{kueng2016robust}) since for
    $\vc{t}=t\cdot\vec(\ID_{\dimPilots})\in\mathbb{R}^{\dimPilots^2}$
    with $t>0$ (giving
    $\lVert \vc{t}\rVert^2_{\ellp{2}}=t^2\dimPilots$) we get the
    strictly (element-wise) positive vector:
  \begin{equation*}
    \vc{w}:=\Amatrix^\herm\vc{t}=t\cdot\Aop^*(\ID_{\dimPilots})=
    \{t\cdot\lVert \avec_k\rVert_{\ellp{2}}^2\}_{k=1}^\dimParam=t\dimPilots\cdot\mathbf{1}>0
  \end{equation*}
  Here we used $\lVert \avec_k\rVert_{\ellp{2}}^2=\dimPilots$ which is
  fulfilled when each $\avec_k$ has unit-magnitude entries.  Since
  this vector is optimally conditioned and allows us to easily use the
  existing results in \cite{kueng2016robust} also for the case
  $p\geq 1$ without extending the proofs (please check here the
  proof of \cite[Theorem 3]{kueng2016robust} for the case where
  $\mt{W}$ is a multiple of the identity - in this case \cite[Lemma
  1]{kueng2016robust} is not necessary).  Assume for now that
  $\Amatrix$ additionally has the $\ell_q$-robust nullspace property
  ($\ell_q$--NSP) of order $s$ with parameters $\rho \in (0,1)$ and
  $\tau >0$ with respect to the $\ell_2$-norm (on
  $\mathbb{C}^{\dimPilots^2}$), meaning that:
  \begin{equation}
    \lVert \vc{v}_S \rVert_{\ellp{q}} \leq \frac{\rho}{s^{1-1/q}} \lVert
    \vc{v}_{\bar{S}} \rVert_{\ellp{1}} + \tau \left\lVert \Amatrix\vc{v} \right\rVert_{\ellp{2}}
    \quad \forall \vc{v}\in\mathbb{R}^\dimParam
    \label{eq:short:def:nsp}
  \end{equation}
  holds for all subsets $S\subset[\dimParam]$ with $|S|\leq s$.  We
  will show below that this indeed the case with high probability for
  $q=2$.  As also well-known (see the discussion yielding \cite[Theorem 4.25]{Foucart2013}),
  $\ell_q$--NSP with respect to a norm $\|\cdot\|$ implies
  $\ell_p$--NSP with respect to $s^{\frac{1}{p}-\frac{1}{q}}\|\cdot\|$
  with the same parameters $(\rho,\tau)$ for $1\leq p\leq q$.

  Note here that, although we consider complex matrices, the vector
  $\vc{v}$ is real-valued.  We conclude then from \cite[Theorem
  3]{kueng2016robust} (and its corresponding straightforward extension
  to $1\leq p\leq q$) for the choice $t=1/\dimPilots$ that (since in
  our case then $\kappa=1$ and, using notation in
  \cite{kueng2016robust}, $\|\mt{W}\|=1/(t\dimPilots)=1$), that the
  solution $\bxnnls$ of the NNLS  
  \eqref{eq_nnls_vec}
  obeys:
  \begin{equation}
    \lVert \bxtrue - \bxnnls\rVert_{\ellp{p}} 
    \leq 
    \frac{2C}{s^{1-\frac{1}{p}}} \sigma_s (\bxtrue)_{\ellp{1}} +
    \frac{2D}{s^{\frac{1}{q}-\frac{1}{p}}} \left( \frac{1}{\dimPilots} +
      \tau \right) \lVert \vc{d} \rVert_{\ellp{2}}  
    \label{eq:short:nsp3}
  \end{equation}
  with $C = \frac{(1+\rho)^2}{1-\rho}$ and
  $D = \frac{3+\rho}{1-\rho}$. This argumentation is already
  sufficient to replace regular $\ell_1$-minimization with NNLS.

  The essential main task here is now to establish that the nullspace
  property for the random matrix $\Amatrix$ holds with high
  probability for the desired sampling rates.  To this end, we will
  restrict to those measurements which are related to the isotropic part
  of $\Amatrix$. More precisely, first define the ``centered''
  matrices $\Acol_k^o:=\avec_k\avec_k^\herm-\ID_{\dimPilots}$.  Now it
  easy to check that if $\avec_k$ has independent and unit magnitude
  entries it follows that for all complex matrices $\mt{Z}$ it holds:
  \begin{equation}
    \begin{split}
      \mathbb{E}|&\langle\Acol^o_k,\mt{Z}\rangle|^2
      =\mathbb{E}|\langle \avec_k,\mt{Z}\avec_k\rangle-\trace(\mt{Z})|^2
      =\lVert \mt{Z}\rVert_\fro^2\\
    \end{split}
    \label{eq:thm:short:variance:strict}
  \end{equation}
  meaning that $\vec(\Acol^o_k)$ is a complex-isotropic random vector.
  Furthermore, this special structure gives us the inequality:
  \begin{equation}
    \lVert\Amatrix\vc{v}\rVert^2_{\ellp{2}}
    =\lVert\Amatrix^o\vc{v}\rVert^2_{\ellp{2}}+\dimPilots|\langle\mathbf{1}_\dimParam,\vc{v}\rangle|^2
    \geq\lVert\Amatrix^o\vc{v}\rVert^2_{\ellp{2}}
    \label{eq:thm:short:Av:lowerbound}
  \end{equation}
  where $\Amatrix^o$ is the corresponding matrix having the
  ``centered'' columns $\vec(\Acol_k^o)$.  Indeed, the inequality
  above is tight since already $2$--sparse vectors can be orthogonal
  to $\mathbf{1}_\dimParam$. This implies, that if the matrix
  $\Amatrix^o$ has $\ell_p$--NSP then also $\Amatrix$ has $\ell_p$--NSP with
  the same parameters.
    
    To establish the NSP of $\Amatrix^o$ we shall make use of existing
    RIP results for the real-valued matrices with independent
    isotropic and heavy-tailed columns.  Therefore, let us consider
    the equivalent real matrix
    $\AmatrixR:=[\Re(\Amatrix^o);\Im(\Amatrix^o)]\in\mathbb{R}^{2\dimPilots^2\times\dimParam}$
    and denote its $k$'th column by
    $\vec(\AcolR_k)\in\mathbb{R}^{2\dimPilots^2}$ where
    $\AcolR_k:=\Re(\Acol^o_k)+i\cdot\Im(\Acol^o_k)$.  For any real
    matrix $\mt{Z}$ we have
  \begin{equation}
    \begin{split}
      \mathbb{E}|\langle \AcolR_k,\mt{Z}\rangle|^2
      &=
      \mathbb{E}(|\Re\langle \Acol^o_k,\mt{Z}\rangle|^2+|\Im\langle \Acol^o_k,\mt{Z}\rangle|^2)\\
      &=\mathbb{E}|\langle \Acol^o_k,\mt{Z}\rangle|^2\overset{\eqref{eq:thm:short:variance:strict}}{=}
      \|\mt{Z}\|^2_\fro
    \end{split}
  \end{equation}
  Thus, the real matrix $\AmatrixR$ has (real-)isotropic and
  independent columns with subexponential marginals.  We adopt a
  normalization $\lambda>0$ such that we have
  $\mathbb{E}\exp(|\langle \lambda\Acol^o_k,\mt{Z}\rangle|)\leq 2$.
  More precisely, from the inequality \cite[Lemma 2.5]{Adamczak2010}
  we have that the subexponential norm
  ($\|\vc{x}\|_{\psi_1}=\inf\{K\geq0\,:\,\mathbb{E}\exp(|\vc{x}|/K)\leq
  2\}$) is controlled by the second moment, i.e., there is a universal
  constant $\lambda$ such that for all $\mt{Z}$ with
  $\|\mt{Z}\|_\fro=1$ we have:
  \begin{equation}
    \lVert \langle \lambda\AcolR_k,\mt{Z}\rangle\rVert_{\psi_1}\leq
    (\mathbb{E}|\langle \AcolR_k,\mt{Z}\rangle|^2)^{1/2}
    =1    \label{eq:short:subexpbound}
  \end{equation}
  Using this normalization, a general $\ell_2$--RIP statement for a matrix
  $\frac{\lambda}{\sqrt{m}}\AmatrixR$ has been shown in
  \cite{Guedon2014:heavy:columns,Adamczak2011}.  Define
  the restricted isometry constant
  \begin{equation}
    \delta_{2s}:=\sup_{0<\|\vc{v}\|_{\ellp{0}}\leq 2s}|\frac{\lVert \lambda\AmatrixR
    \vc{v} \rVert^2_{\ellp{2}}}{m\|\vc{v}\|^2_{\ellp{2}}}-1|
  \end{equation}
  If $\delta_{2s}\in[0,1)$ the matrix
  $\frac{\lambda}{\sqrt{m}}\AmatrixR$ has $\ell_2$--RIP of order $2s$.
  Furthermore, if even $\delta_{2s}<4/\sqrt{41}\approx0.62$ then
  $\frac{\lambda}{\sqrt{m}}\AmatrixR$ has the $\ell_2$--robust NSP of
  order $s$ with parameters with
  $\rho=\delta_{2s}/(\sqrt{1-\delta_{2s}^2}-\delta_{2s}/4)$ and
  $\tau'=\sqrt{1+\delta_{2s}}/(\sqrt{1-\delta^2_{2s}}-\delta_{2s}/4)$
  \cite[Theorem 6.13]{Foucart2013}.  We use now \cite[Theorem 1, case
  2, $\alpha=1$]{Guedon2014:heavy:columns}.  Assume
  $8/\delta\leq \dimParam\leq c s\exp(c s\sqrt{m}/2)$ (for example take
  $\dimParam\geq 16$ to have some feasible $\delta\leq4/\sqrt{41}$;
  for the constant $c$, see \cite{Adamczak2011}))
  and set
  the sparsity parameter to the following integer part:
  \begin{equation}
    2s=[\frac{\delta^2}{c'}\log[\frac{c'\dimParam}{\delta^2m}]^{-2}]
  \end{equation}
  then, according to mention result in \cite{Guedon2014:heavy:columns}
  it holds that
  $\mathbb{P}\{\delta_{2s}\leq\delta\}\geq 1-2^{-9}\delta$ ($c'$ is
  the constant $C^2$ in \cite{Guedon2014:heavy:columns}).  Furthermore
  for $m \geq c' s/\delta^2$ this can be rearranged
  to:
  \begin{equation}
    \dimPilots(\dimPilots-1)=m\geq c'\delta^{-2}s\log^2(\dimParam/s)
  \end{equation}
  Therefore for all $\vc{v}\in\mathbb{R}^\dimParam$
  and all subsets $S\subset[\dimParam]$ with $|S|\leq s$ is holds:
  \begin{equation}
      \begin{split}
        \lVert \vc{v}_S \rVert_{\ellp{2}}
        &\leq
        \frac{\rho}{\sqrt{s}} \lVert
        \vc{v}_{\bar{S}} \rVert_{\ellp{1}} + \frac{\tau' \lambda}{\sqrt{m}}\left\lVert \AmatrixR\vc{v} \right\rVert_{\ellp{2}}
        \\
        &=
        \frac{\rho}{\sqrt{s}} \lVert
        \vc{v}_{\bar{S}} \rVert_{\ellp{1}} + \frac{\tau' \lambda}{\sqrt{m}}\left\lVert \Amatrix^o\vc{v} \right\rVert_{\ellp{2}}
        \\
        &\overset{\eqref{eq:thm:short:Av:lowerbound}}{\leq}
        \frac{\rho}{\sqrt{s}} \lVert
        \vc{v}_{\bar{S}} \rVert_{\ellp{1}} + \frac{\tau' \lambda}{\sqrt{m}}\left\lVert \Amatrix\vc{v} \right\rVert_{\ellp{2}}
      \end{split}
    \end{equation}
    with $\rho\leq\delta/(\sqrt{1-\delta^2}-\delta/4)$ and
    $\tau'\leq\sqrt{1+\delta}/(\sqrt{1-\delta^2}-\delta/4)$.
    Summarizing, $\ell_2$--RIP of order $2s$ for
    $\frac{\lambda}{\sqrt{m}}\AmatrixR$ implies $\ell_2$--NSP of
    $\Amatrix$ of order $s$ with parameters $\rho$ and
    $\tau=\tau'\lambda/\sqrt{m}$.

    Finally, let us write \eqref{eq:short:nsp3} for the case $q=2$
    (therefore $1\leq p\leq 2$) in a more
    convinient form since $m=\dimPilots(\dimPilots-1)$:
    \begin{equation}
      \begin{split}
        \lVert \bxtrue &- \bxnnls\rVert_{\ellp{p}} 
        \overset{\eqref{eq:short:nsp3}}{\leq}
        \frac{2C}{s^{1-\frac{1}{p}}} \sigma_s (\bxtrue)_{\ellp{1}} +
        \frac{2D}{s^{\frac{1}{2}-\frac{1}{p}}} \left( 1 +
          \dimPilots\tau \right) \frac{\lVert \vc{d} \rVert_{\ellp{2}}}{\dimPilots}\\
        &=
        \frac{2C}{s^{1-\frac{1}{p}}} \sigma_s (\bxtrue)_{\ellp{1}} +
        \frac{2D}{s^{\frac{1}{2}-\frac{1}{p}}} \left( 1 +
          \frac{\lambda\tau'\sqrt{\dimPilots}}{\sqrt{\dimPilots-1}} \right) \frac{\lVert \vc{d} \rVert_{\ellp{2}}}{\dimPilots}
      \end{split}
      \label{eq:short:nsp4}
    \end{equation}
    which is \eqref{eq:short:thm:mse}.

\subsection{Analysis of Error of the Sample Covariance Matrix}\label{cov_err_app}
We first consider the simple case where $\{\bfy(t): t \in [M]\}$, are $M$ i.i.d. realization of a $D_c$-dim complex Gaussian vector with zero mean $\bE[\bfw(t)]={\bf 0}$ and a diagonal covariance matrix $\Sigmam_\bfy=\bE[\bfy(t) \bfy(t)^\herm]=\diag(\betam)$, thus, $y_i(t) \sim \cg(0, \beta_i)$, $i\in[D_c]$. Let us denote by $\Deltam=\widehat{\Sigmam}_\bfy -  \Sigmam_\bfy$ be the deviation of the sample covariance matrix from its mean. Note that the $(i,j)$ component of $\Deltam$ is given by
\begin{align}
\Delta_{ij}=\frac{1}{M} \sum_{t \in [M]} y_i(t)y_j^*(t) - \beta_i \delta_{ij}
\end{align}
where $\delta_{ij}=\bI_{\{i=j\}}$ denotes the discrete delta function. Note that $\bE[\Delta_{ij}]=0$ for all $i,j$. Moreover, since $\Delta_{i,j}$ is the average of $M$ i.i.d. terms, $\{y_i(t)y_j^*(t) - \beta_i \delta_{ij}: t \in [M]\}$, we have  
\begin{align}
\bE[|\Delta_{ij}|^2]=\frac{\bE[|y_i(t)y_j^*(t) - \beta_i \delta_{ij}|^2]}{M}.
\end{align}
For $i \not = j$, we have that  
\begin{align}
\bE[|y_i(t) y_j(t)^* - \beta_i \delta_{i,j}|^2]&=\bE[|y_i(t) y_j(t)^*|^2]\nonumber\\
&\stackrel{(a)}{=}\bE[|y_i(t)|^2] \bE[|y_j(t)|^2]\nonumber\\
&=\beta_i \beta_j,\label{cov_err_dumm_1}
\end{align}
where in $(a)$ we used the independence of the different components of $\bfy(t)$.
Also, for $i=j$, we have that
\begin{align}
\bE[\big | y_i(t)y_j(t)^*&-\beta_i \delta_{i,j} \big |^2]=\bE[\big | |y_i(t)|^2-\beta_i  \big |^2]\nonumber\\
&=\bE[|y_i(t)|^4] -2\beta_i \bE[|y_i(t)|^2]+ \beta_i^2\nonumber\\
&\stackrel{(a)}{=}2\bE[|y_i(t)|^2]^2 -2\beta_i^2+\beta_i^2\nonumber\\
&=2\beta_i^2- 2\beta_i^2 + \beta_i^2=\beta_i^2,\label{cov_err_dumm_2}
\end{align}
where in $(a)$ we used the identity $\bE[|y_i(t)|^4]=2\bE[|y_i(t)|^2]^2$ for complex Gaussian random variables.
Overall, from \eqref{cov_err_dumm_1} and \eqref{cov_err_dumm_2}, we can write $\bE[|\Delta_{ij}|^2]=\frac{\beta_i \beta_j}{M}$. 
Thus, we have that
\begin{align}
\bE[\|\Deltam\|_{\sfF}^2]=\sum_{ij} \bE[|\Delta_{ij}|^2] &= \frac{\sum_{i,j} \beta_i \beta_j}{M}\nonumber\\
&=  \frac{(\sum \beta_i)^2}{M} = \frac{\tr(\Sigmam_\bfy)^2}{M}.\label{exp_del}
\end{align}

\begin{remark}
It is worthwhile to mention that although  \eqref{exp_del} was derived under the Gaussianity of the observations $\{\bfy(t): t\in [M]\}$, the result can be easily modified for general distribution of the components of $\bfy(t)$. More specifically, let us define  
\begin{align}\label{eta_bound}
\max_{i} \frac{\bE[|y_i(t)|^4]}{\bE[|y_i(t)|^2]^2}=:\varsigma < \infty.
\end{align}
Then, using \eqref{cov_err_dumm_1} and applying \eqref{eta_bound} to \eqref{cov_err_dumm_2}, we can obtain the following upper bound
\begin{align}
\bE[\|\Deltam\|_{\sfF}^2] &\leq \max\{\varsigma-1,1\} \times \frac{\sum_{i,j} \beta_i \beta_j }{M}\\
& \leq \max\{\varsigma-1,1\} \times \frac{\tr(\Sigmam_\bfy)^2}{M},
\end{align}
which is equivalent to \eqref{exp_del} up to the constant multiplicative factor $\max\{\varsigma-1,1\}$. \hfill $\lozenge$
\end{remark}

In practice, $\|\Deltam\|_\sfF^2$ concentrates very well around its mean $\bE[\|\Deltam\|_{\sfF}^2]$. Therefore, 
the deviation  between the true and empirical covariance matrix can be approximated by
\begin{align}\label{cov_err_dumm_3}
\|\widehat{\Sigmam}_{\bfy} - \Sigmam_\bfy\|_\sfF \simeq \frac{ \tr(\Sigmam_\bfy)}{\sqrt{M}}.
\end{align}

Now, assume that the covariance matrix $\Sigmam_\bfy$ is not in a diagonal form and let $\Sigmam_\bfy=\bfU \diag(\betam) \bfU^\herm$ be the Singular Value Decomposition (SVD) of $\Sigmam_{\bfy}$. By multiplying all the vectors $\bfy(t)$ by the orthogonal matrix $\bfU^\herm$ to whiten them and noting the fact that multiplying by $\bfU^\herm$ does not change the Frobenius norm of a matrix, we can see that  \eqref{cov_err_dumm_3} holds true in general also for non-diagonal covariance matrices.


\section{Acknowledgement} 
The authors would like to thank R. Kueng for inspiring discussions and helpful comments. P.J. is supported by DFG grant JU 2795/3.

\balance 

{\small
\bibliographystyle{IEEEtran}
\bibliography{references2}
}

\end{document}